%% file: main.tex
\documentclass[11pt]{article}

\usepackage{amsfonts,amsmath,amsthm,amssymb}
\usepackage[dvipsnames]{xcolor}
\usepackage{float}
\definecolor{DarkBlue}{rgb}{0.1,0.1,0.5}
\usepackage[colorlinks,citecolor=ForestGreen,linkcolor=blue,urlcolor=DarkBlue,pagebackref]{hyperref}
\usepackage{xspace}
\usepackage[margin=1in]{geometry}
\usepackage{graphicx} 
\usepackage[capitalise,nameinlink]{cleveref}
\usepackage{thm-restate}
\usepackage{framed}
\usepackage{tikz}
\usetikzlibrary{backgrounds,decorations.pathreplacing,fadings,calc,positioning}
\usepackage{float}
\usepackage{multirow,makecell}
\usepackage{stmaryrd}
\usepackage{url}
\usepackage{paralist}

\usepackage{tcolorbox}
\usepackage{footnote}
\BeforeBeginEnvironment{tcolorbox}{\begin{savenotes}}
\AfterEndEnvironment{tcolorbox}{\end{savenotes}}

\usepackage{algorithm, algpseudocode}
\algtext*{EndWhile} 
\algtext*{EndIf} 
\algtext*{EndFor} 
\algtext*{EndFunction} 

\definecolor{midnight}{rgb}{0.0, 0.2, 0.4}
\definecolor{dblue}{rgb}{0.12, 0.56, 1.0}
\definecolor{amber}{rgb}{1.0, 0.49, 0.0}
\definecolor{dorange}{rgb}{1.0, 0.55, 0.0}
\definecolor{dtang}{rgb}{1.0, 0.66, 0.07}
\definecolor{dsblue}{rgb}{0.0, 0.75, 1.0}
\definecolor{denim}{rgb}{0.08, 0.38, 0.74}
\definecolor{pblue}{rgb}{0.2, 0.2, 0.6}
\definecolor{gblue}{rgb}{0.0, 0.58, 0.71}
\definecolor{cdblue}{rgb}{0.16, 0.32, 0.75}
\definecolor{dlavender}{rgb}{0.45, 0.31, 0.59}
\definecolor{alizarin}{rgb}{0.82, 0.1, 0.26}
\definecolor{dorchid}{rgb}{0.6, 0.2, 0.8}
\definecolor{dimgray}{rgb}{0.41, 0.41, 0.41}
\definecolor{skobeloff}{rgb}{0.0, 0.48, 0.45}
\definecolor{cgreen}{rgb}{0.0, 0.8, 0.6}

\tikzset{
    ncbar angle/.initial=90,
    ncbar/.style={
        to path=(\tikztostart)
        -- ($(\tikztostart)!#1!\pgfkeysvalueof{/tikz/ncbar angle}:(\tikztotarget)$)
        -- ($(\tikztotarget)!($(\tikztostart)!#1!\pgfkeysvalueof{/tikz/ncbar angle}:(\tikztotarget)$)!\pgfkeysvalueof{/tikz/ncbar angle}:(\tikztostart)$)
        -- (\tikztotarget)
    },
    ncbar/.default=0.1cm,
}
\tikzset{square left brace/.style={ncbar=0.1cm}}
\tikzset{square right brace/.style={ncbar=-0.1cm}}

\newcommand{\set}[1]{\left\{#1\right\}}

\DeclareMathOperator{\poly}{poly}

\usepackage{old-arrows}

\renewcommand{\epsilon}{\varepsilon}

\newcommand{\BN}{\mathbb{N}}
\newcommand{\be}{\mathbf{e}}
\newcommand{\bu}{\mathbf{u}}
\newcommand{\bv}{\mathbf{v}}
\newcommand{\cB}{\mathcal{B}}
\newcommand{\cC}{\mathcal{C}}
\newcommand{\cO}{\mathcal{O}}
\newcommand{\cP}{\mathcal{P}}

\newcommand{\cV}{\mathcal{V}}
\newcommand{\sI}{\mathsf{I}}
\newcommand{\sO}{\mathsf{O}}
\newcommand{\dist}{\mathsf{dist}}

\theoremstyle{definition}
\newtheorem{theorem}{Theorem}[section]

\newtheorem{lemma}[theorem]{Lemma}

\newtheorem{corollary}[theorem]{Corollary}
\newtheorem{definition}[theorem]{Definition}

\newtheorem{remark}[theorem]{Remark}

\crefname{claim}{claim}{claims}

\NewDocumentEnvironment{claimproof}{ }{%
  \proof
}{%
  \endproof
}

\definecolor{amethyst}{rgb}{0.6, 0.4, 0.8}
\newcommand{\revised}[1]{#1}

\title{Better Diameter Bounds for Efficient Shortcuts\\ and a Structural Criterion for Constructiveness}

\author{Bernhard Haeupler\thanks{INSAIT, Sofia University ``St. Kliment Ohridski'' and ETH Zurich. \href{mailto:bernhard.haeupler@inf.ethz.ch}{bernhard.haeupler@inf.ethz.ch}. This research was partially funded by the Ministry of Education and Science of Bulgaria (support for INSAIT, part of the Bulgarian National Roadmap for Research Infrastructure) and by the European Research Council (ERC) under the European Union's Horizon 2020 research and innovation program (ERC grant agreement 949272).} \and
Antti Roeyskoe\thanks{ETH Zurich. \href{mailto:antti.roeyskoe@inf.ethz.ch}{antti.roeyskoe@inf.ethz.ch}. This research was partially funded by the European Research Council (ERC) under the European Union's Horizon 2020 research and innovation program (ERC grant agreement 949272).} \and
Zhijun Zhang\thanks{INSAIT, Sofia University ``St. Kliment Ohridski''. \href{mailto:zhijun.zhang@insait.ai}{zhijun.zhang@insait.ai}. This research was partially funded by the Ministry of Education and Science of Bulgaria (support for INSAIT, part of the Bulgarian National Roadmap for Research Infrastructure).}}
\date{}

\begin{document}

\pagenumbering{gobble}

\maketitle

\begin{abstract}
\revised{%
All parallel algorithms for directed reachability and shortest paths crucially rely on efficient shortcut constructions. These constructions find directed paths and shortcut them by adding edges, with the goal to reduce the diameter of the graph. A long sequence of works has studied (efficient) shortcut constructions as well as impossibility results on the best diameter and therefore the best parallelism that can be achieved via this approach.\\

This paper introduces a new conceptual 
tool for this line of research in the form of a simple and natural structural criterion: A shortcut $H$ for a graph $G$ is \emph{certified} if for any shortcut edge $(u, v) \in H$, there exists a vertex $w$ such that the edges $(u, w)$ and $(w, v)$ are also in $G \cup H$.\\

We show that this criterion captures constructiveness in the following sense: A shortcut $H$ can be constructed in $t$ time by repeatedly spending $\ell$ time on shortcutting a path of length $\ell$, if and only if, there exists a certified shortcut $H' \supseteq H$ of size $\tilde{O}(t)$. Furthermore, all known shortcut constructions with efficient algorithms can be extended to produce certified shortcuts of size $\tilde{O}(m)$. On the other hand, for shortcut constructions for which attempts to find efficient implementations have failed, we can show that this is impossible.\\

We also obtain stronger diameter lower bounds for certified shortcuts and hopsets. For example, no certified shortcut construction with almost-linear size can reduce a graph's diameter below $n^{1/4-o(1)}$. This seems to be the best bound one can hope for with current techniques.
}
\end{abstract}

\newpage

\tableofcontents

\newpage

\pagenumbering{arabic}

\input{intro}
\input{prelim}
\input{lb}
\input{existing}

\paragraph{Acknowledgement.}
We are thankful to Seth Pettie and the anonymous reviewers of ICALP 2026 for their helpful comments on the paper.

\bibliographystyle{alphaurl}
\bibliography{main}

\appendix

\input{appendix_lowdepth}
\input{app_proof}
\input{appendix_treaps}

\end{document}

%% file: intro.tex
\section{Introduction}
\label{sec:intro}

One of the central and long-standing open problems in parallel computing is computing shortest paths or even just directed reachability.
Despite the presence of simple and efficient algorithms in the sequential setting, designing parallel algorithms that are both work-efficient and have low depth is notoriously difficult.
So far, the best such algorithm, due to \cite{JambulapatiLS19}, has depth $n^{1/2+o(1)}$.

The critical ingredient and primary idea powering all known parallel algorithms are \emph{shortcuts}.
On a directed acyclic graph (DAG) $G = (V, E)$ of diameter $D$, a standard parallel breadth-first search algorithm can resolve reachability with almost-linear work and depth $\tilde{\cO}(D)$. However, the input DAG could have a diameter as large as $\Theta(n)$. Computing a shortcut to reduce the diameter can drastically improve the depth of the breadth-first search. 
\begin{definition}
    A shortcut $H$ for a DAG $G$ is a subset of edges from
    the transitive closure of $G$.
\end{definition}
Suppose that there exists a shortcut $H$ of size $|H| = m^{1 + o(1)}$ that can be constructed with almost-linear work and depth $D'$, such that, the graph $G \cup H$ has diameter $D$. This immediately implies an efficient algorithm for directed reachability with almost-linear work and depth $\tilde{\cO}(D + D')$ --- simply compute the shortcut $H$ and run a breadth-first search on $G \cup H$.
Conversely, if no (efficiently computable) shortcuts with low diameters exist in general graphs, then this poses a limitation on the best depth achievable by any algorithm following this shortcut-based approach, and therefore a severe barrier for improving the current state of the art. 

A major advantage of studying shortcuts is that they are purely combinatorial structures, independent of any computation model. This, combined with the above connection, makes it possible to apply graph-theoretic and purely structural results to give strong evidence about parallel computations.

A long line of work has studied the (non-)existence of shortcuts with low diameters. The state of the art lower bound due to \cite{HoppenworthXX25} shows that $\cO(m)$-size shortcuts cannot in general reduce the diameter of a graph below $\Omega(n^{2/9})$, while
on the upper bound side, $\cO(n)$-size shortcuts reducing diameter to $\cO(n^{1/3})$ are known to exist for all graphs \cite{KoganP22,bals2025greedy}.

Similar approaches have also been extended to \emph{hopsets}, which are generalizations of shortcuts that reduce the number of edges on \emph{shortest} paths, instead of arbitrary paths, while preserving distances, not just connectivity.
\begin{definition}
    A hopset $H$ for a graph $G$ with positive edge lengths is a set of edges such that $\dist_{G}(u, v) = \dist_{G \cup H}(u, v)$ for every pair of vertices $u, v$. 
\end{definition}
A hopset $H$ for $G$ is said to have \textit{hopbound} $\beta$ if for any pair of vertices $u, v$, where $u$ can reach $v$, $G \cup H$ contains a shortest $(u, v)$-path with at most $\beta$ edges. So far, the best known hopset construction remains the one produced by the folklore algorithm \cite{UllmanY91,BermanRR10}, obtaining a hopbound of $\cO(n^{1 / 2})$ by adding $\cO(n)$ edges. For lower bounds, \cite{HoppenworthXX25} shows $\cO(m)$-size hopsets cannot in general obtain a hopbound of $o(n^{2/7})$.

A more detailed discussion on prior work is provided in \cref{sec:related}.

\subsection{Our Results}
\label{sec:result}

This paper introduces new conceptual tools to study the structure of shortcuts that can be \emph{efficiently} produced by a large and natural family of shortcut construction algorithms.

There are two aspects in which existential results on shortcuts are non-constructive. Firstly, how does one efficiently decide which node pairs to shortcut? Secondly, how can one be convinced that an edge $(u,v)$ is a valid shortcut edge, i.e., how can one verify that there is a $(u,v)$-path? We focus on the second aspect and define the following \emph{certified shortcutting procedure}, in which the algorithm still gets (nondeterministic) advise on which shortcut edges to add and also on paths that certify the validity of these edges, but forces the algorithm to verify the path:

\begin{definition}
A \emph{certified shortcutting procedure} for a graph $G$ starts with an empty shortcut $H$ and repeatedly takes some path $P=(u_1, \ldots,u_{|P|})$ in the \emph{current} graph $G \cup H$, verifies it using $|P|$ steps, and then adds to $H$ an arbitrary number of shortcut edges $(u_i,u_j)$ with $i<j$ using one step per edge.  
\end{definition}

We note that shortcut edges can be used in subsequent paths and therefore speed up future verifications. For example, if one wants to
shortcut two paths $P$ and $P'$ that share a long subpath $P''$, then they can be verified in $|P''| + (|P \setminus P''|+1)+ (|P' \setminus P''|+1) = |P \cup P'|+2$ steps by first adding an edge for $P''$.

We call an algorithm for directed reachability \emph{certified-shortcut-based} if it first constructs a shortcut $H$ by a certified shortcutting procedure and then runs a breadth-first search on $G \cup H$. The first main result of this paper gives a stronger lower bound on the best depth achievable by the certified-shortcut-based framework.

\begin{theorem}
\label{thm:reach}
Any certified-shortcut-based algorithm for directed reachability with $\tilde{O}(m)$ work must have depth $\tilde{\Omega}(n^{1/4})$.
\end{theorem}

\subsubsection{Certified Shortcuts: Definition and Equivalences}

The key conceptual ingredient that powers the proof of \cref{thm:reach} is a simple and natural structural criterion, which captures shortcuts that can be easily verified and are therefore more constructive in some sense.

\begin{definition}
\label{def:cert-sc}
A shortcut $H$ is certified if for any edge $(u,v) \in H$, there exist two edges $(u,w),(w,v) \in E \cup H$ for some vertex $w \in V$.
\end{definition}

Being certifiable is a static structural property of a shortcut which is easy to state and easy to (formally) argue about. One way to interpret this criterion is to think of it as a requirement that every edge in $H$ has two ''parent'' edges from which it could have been derived. This induces a natural order in which any certified shortcut can be constructed via a particularly simple certified shortcutting procedure that repeatedly shortcuts only paths of length two. We call such a procedure a \emph{$2$-hop certified shortcutting procedure}.

\begin{restatable}[]{lemma}{lemequi}
\label{lem:equi-seq-sc}
Any shortcut $H$ for a DAG is certified if and only if it can be constructed by a $2$-hop certified shortcutting procedure. 
\end{restatable}

\begin{proof}
Shortcuts constructed by $2$-hop certified shortcutting procedures are certified by definition.
For the other direction, since $H$ is a shortcut for a DAG $G$, $G \cup H$ is also a DAG. Fix a topological order $(v_1, v_2, \dots, v_n)$ of $G \cup H$. Consider some shortcut edge $(v_i, v_j) \in H$ certified by $(v_i, v_k), (v_k, v_j) \in G \cup H$. It must hold that $i < k < j$, otherwise the topological order would not be valid. Thus, a $2$-hop certified shortcutting procedure can insert the edges $(v_i, v_j) \in H$ in nondecreasing order of $j - i$, as $j - i > \max \{j - k, k - i\}$.
\end{proof}

A similar equivalence between general certified shortcutting procedures and certified shortcuts holds, up to $\tilde{O}(1)$ factors. In particular, one can transform any certified shortcutting procedure into a $2$-hop one with small overhead:

\begin{restatable}[]{lemma}{lemtwosteps}
\label{lem:lem-two-steps}
If a shortcut $H$ for a DAG $G$ can be constructed by a certified shortcutting procedure in $t$ steps, then there exists a certified shortcut $H' \supseteq H$ that can be constructed by a $2$-hop certified shortcutting procedure using $\tilde{O}(t)$ steps.
\end{restatable}

The following corollary follows directly from the two lemma above.

\begin{corollary}
\label{corr}
 If for a graph $G$, there does not exist a certified shortcut of size $t$ that reduces the diameter to $D$, then any certified shortcutting procedure that reduces the diameter of $G$ to at most $D$ requires at least $\tilde{\Omega}(t)$ steps. 
\end{corollary}

On the other hand, if there exists a certified shortcut $H$ of size $t$ with $\mathrm{diam}(G \cup H) \leq D$ then there exists a simple (nondeterministic) $2$-hop certified shortcutting procedure that takes $t$ steps and reduces the diameter to $D$ while also producing an easily and locally checkable certificate for the validity of $H$. 
 
\subsubsection{Lower Bounds on Certified Shortcuts and Hopsets}

With the additional requirement of a shortcut being certified, we are able to prove stronger diameter lower bounds.

\begin{theorem}
\label{thm:lb-sc-m}
There exists a family of DAGs with $n$ vertices and $m$ edges, such that any certified shortcut of size $\tilde{O}(m)$ has diameter $\tilde{\Omega}(n^{1/4})$.
\end{theorem}

This theorem together with \cref{corr} completes the proof of \cref{thm:reach}.
We remark that \cref{thm:lb-sc-m} can be easily derived from the constructions of \cite{Hesse03,HuangP21} and an improved alternation product due to \cite{LuWWX22}.

Our notion of certifiability also naturally extends to (exact) hopsets, which generalize shortcuts, by additionally requiring the weight of an added edge $(u,v)$ to be equal to the total weights of its two certifying edges $(u,w),(w,v)$.
We similarly obtain improved lower bounds on hopbounds of certified hopsets.

\begin{theorem}
\label{thm:lb-hs-m}
There exists a family of unweighted undirected graphs with $n$ vertices and $m$ edges, such that any certified hopset of size $\tilde{O}(m)$ has hopbound $\tilde{\Omega}(n^{1/3})$.
\end{theorem}

To put this bound in context, we remark that existing hopset lower bound results had a gap between the cases of
unweighted and weighted graphs.
Though not explicitly stated, it appears that the lower bound of \cite{BodwinH23}, combined with alternation products, implies a lower bound of $\tilde{\Omega}(n^{1/3})$ for the hopbound of $\tilde{O}(m)$-size hopsets in weighted graphs but left the case of unweighted graphs open.
Subsequently, \cite{HoppenworthXX25} proved a lower bound of $\tilde{\Omega}(n^{2/7})$ for unweighted graphs. The lower bounds in \cref{thm:lb-hs-m} proves the $\tilde{\Omega}(n^{1/3})$ bound for certified shortcuts in unweighted graphs.

The discussion so far has focused on the setting where  $\tilde{O}(m)$ edges are added, which are the largest amount of edges any efficient algorithm can afford to add. For completeness, we also give lower bounds on certified shortcuts and hopsets with $\tilde{O}(n)$ edges, even though these are less natural in the context of certified shortcuts. 

\begin{theorem}
\label{thm:lb-sc-n}
There exists a family of DAGs with $n$ vertices, such that any certified shortcut of size $\tilde{O}(n)$ has diameter $\tilde{\Omega}(n^{1/3})$.
\end{theorem}

\begin{theorem}
\label{thm:lb-hs-n}
There exists a family of unweighted undirected graphs with $n$ vertices, such that any certified hopset of size $\tilde{O}(n)$ has hopbound $\tilde{\Omega}(n^{1/2})$.
\end{theorem}

\subsubsection{Certification Complexity}

As a final main contribution, this paper shows that one can use the concept of certified shortcuts to define a complexity measure $\mathcal{C}(H,G)$ that gives some indication on how fast one can hope to construct $H$. 
This complexity measure $\mathcal{C}(H,G)$ is purely structural, i.e., it does not involve and is independent of the procedure or algorithm that produced $H$. It is also simple enough that it can be argued about formally. We believe this concrete measure is a useful indicator and provides valuable guidance towards assessing whether any given shortcut construction might have an efficient implementation. We give some evidence for this belief by formally proving that our complexity measure perfectly predicts or at least coincides with the current landscape of results on the efficiency of all known shortcut constructions~\cite{UllmanY91,BermanRR10,Fineman18,JambulapatiLS19,KoganP22,KoganP23,bals2025greedy}. 

Concretely, we define the \emph{certification complexity} $\mathcal{C}(H,G)$ of a shortcut $H$ in a DAG $G$ to be the smallest size of any certified shortcut $H'$ in $G$ that contains $H$:

\begin{definition}
\label{def:cert-comp-sc}
$\mathcal{C}(H,G) = \min\left\{|H'| \ \mid \ H' \supseteq H\ \text{is a certified shortcut in }G\right\}.$
\end{definition}

The name certification complexity comes from the following observation: A certified shortcut $H'$ that contains a given shortcut $H$ can be seen as a natural and locally checkable certificate for the validity of $H$. In particular, it suffices for any edge in a certified shortcut $H'$ to look at only its adjacent edges in $G \cup H'$ to detect a violation. The certification complexity of a shortcut $H$ is therefore exactly the size of the smallest locally checkable certificate $H'$ in this form that proves the validity of $H$. 

Alternatively, using the equivalences explained in \cref{lem:equi-seq-sc} and \cref{lem:lem-two-steps}, one can also think of $\mathcal{C}(H,G)$ as exactly the number of steps a certified shortcutting procedure requires to construct (a super set of) $H$, up $\tilde{O}(1)$ factors.

It is now natural to informally say that a shortcut $H$ in a graph $G$ is $\mathcal{C}$-indicated to be ''efficiently constructible'' if its certification complexity $\mathcal{C}(H,G)$ is comparable to the size of $G$ and $\mathcal{C}$-indicated to be ''not efficiently constructible'' otherwise. Similarly, any (randomized) shortcut construction, which given a graph $G$ defines a shortcut $H(G)$ (or a distribution over such shortcuts), can be formally classified based on the complexity measure $\mathcal{C}$ as follows:

We call a construction \emph{certified-shortcut efficient} if for every graph $G$, the construction (with high probability) produces shortcuts with certification complexity $\tilde{O}(m)$. Conversely, we call a construction \emph{provably certified-shortcut inefficient} if there exists a family of graphs such that the construction with high probability produces a shortcut with certification complexity $m^{1+\Omega(1)}$ for every graph in this family. 

We prove in \cref{sec:existing} that this classification perfectly matches the state of the art regarding what shortcut constructions have a known efficient implementation. 

\begin{theorem}[Informal]
\label{thm:existing}
Every known shortcut construction that has a known sequential/parallel algorithm with near-linear time/work~\cite{Fineman18,JambulapatiLS19,KoganP23,bals2025greedy} is certified-shortcut efficient, i.e., provably always constructs shortcuts with certification complexity $\tilde{O}(m)$.

Conversely, every known shortcut construction for which no  efficient implementation is known, i.e., \cite{UllmanY91,KoganP22,BermanRR10,bals2025greedy}, is provably certified-shortcut inefficient. 
\end{theorem}

\subsection{Prior Work}
\label{sec:related}

\paragraph{Shortcuts.}

The study of shortcuts dates back to \cite{UllmanY91}, which provides a simple argument, based on vertex sampling, showing the existence of $\cO(n^{1/2})$-shortcuts of size $\tilde{\cO}(n)$.
This folklore proof remains essentially the best known upper bound for decades,
until a major breakthrough due to \cite{KoganP22} significantly improves the upper bound by constructing $\cO(n^{1/3})$-shortcuts of size $\tilde{\cO}(n)$.
Their result crucially relies on sampling of both vertices and paths.

The use of random sampling is first avoided using greedy approaches in \cite{BermanRR10}, which eliminates the logarithmic factors in the guarantees of \cite{UllmanY91}.
Very recently, building upon this, \cite{bals2025greedy} presents a similar improvement on the logarithmic factors over \cite{KoganP22} also using greedy approaches.
In fact, they further give a much more sophisticated analysis of \cite{BermanRR10} showing that it actually achieves the same guarantee of \cite{KoganP22} up to logarithmic factors.

Substantial effort has been made to speed up the computation of shortcuts.
Specifically, improving upon \cite{KoganP22fast}, \cite{KoganP23,bals2025greedy} demonstrate almost-linear-time algorithms for computing $\cO(n^{1/2})$-shortcuts using flow algorithms.

In another concurrent line of work, \cite{Fineman18} adapts the vertex-based approach of \cite{UllmanY91} and designs a work-efficient parallel algorithm for computing shortcuts in $\tilde{\cO}(n^{2/3})$ depth.
The depth and diameter bound is later improved to $n^{1/2+o(1)}$ by \cite{JambulapatiLS19}.
Unfortunately, it is unclear how to extend these frameworks to efficiently implement the path-based approach of \cite{KoganP22}.
In fact, assuming standard fine-grained complexity assumptions, any straightforward implementation would require strictly superlinear time even in the sequential setting \cite{KoganP23}.

On the lower bound side, \cite{Hesse03} constructs a family of layered DAGs, the diameter of which cannot be reduced below $\Omega(n^{1/17})$ by any $\cO(m)$-size shortcuts.
In particular, it refutes a conjecture by \cite{Thorup92} asserting the existence of linear-size shortcuts with polylogarithmic diameters for general graphs.
Later on, \cite{HuangP21} and subsequently \cite{LuWWX22} strengthen the diameter lower bound to $\Omega(n^{1/8})$ via better alternation products.
All these constructions work with grid points in low dimensional Euclidean spaces and rely on a classical result by \cite{barany1998convex} which guarantees the existence of large convex sets.
Currently, the state-of-the-art $\Omega(n^{2/9})$ lower bound is proved by \cite{HoppenworthXX25} using new constructions without the use of such convex sets.

Regarding $\cO(n)$-size shortcuts, the best known diameter lower bound is $\tilde{\Omega}(n^{1/4})$ due to \cite{BodwinH23}, building upon \cite{HuangP21}.
\cite{WilliamsXX24} reproduces (and slightly improves by eliminating the logarithmic factors) the same result using constructions that later inspired \cite{HoppenworthXX25}.

\paragraph{Hopsets.}

As hopsets are also shortcuts, all lower bounds for shortcuts naturally extend to hopsets.
Additionally, \cite{HoppenworthXX25} shows that no $\cO(m)$-size hopsets can reduce hopbounds below $\Omega(n^{2/7})$, an improvement over the $\Omega(n^{2/9})$ lower bound for shortcuts.
The proof is by observing that uniqueness of shortest paths, rather than arbitrary paths, between critical pairs is sufficient.

On the other hand, edge sets constructed by shortcut algorithms are not necessarily hopsets.
In fact, the folklore sampling algorithm \cite{UllmanY91} and the original greedy approach \cite{BermanRR10} turn out to be essentially the only known shortcut algorithms, regardless of efficiency, that also produce hopsets.
The same guarantee of $\cO(n^{1/2})$ hopbound holds.

It is also worth noting that for hopsets, whether or not the given graph is weighted also makes a difference, unlike in the case of shortcuts.
While all aforementioned lower bound constructions are unweighted, \cite{BodwinH23} shows that \cite{UllmanY91,BermanRR10}, which also work on weighted graphs, are essentially optimal for weighted graphs.
Specifically, there exists a family of weighted graphs for which hopbounds cannot be reduced below $\tilde{\Omega}(n^{1/2})$ using any $\cO(n)$-size hopset.
Though not explicitly stated, via alternation products, it also appears to imply an $\tilde{\Omega}(n^{1/3})$ hopbound lower bound for $\cO(m)$-size hopsets.
Very recently, \cite{bals2025greedy} proves that the greedy approach of \cite{BermanRR10} is existentially optimal in terms of both $n$ and $m$ for hopsets in weighted graphs.

A summary of our new results and best prior bounds can be found in \cref{tab:lb,tab:ub}.

\begin{table}[tp]
    \centering
    \begin{tabular}{c|c|c|c}
        \hline
         & $|H|$ & $D$ & \\
        \hline
        \multirowcell{4}{Shortcut} & \multirowcell{2}{$\cO(n)$} & $\Omega(n^{1/4})$ & \cite{BodwinH23,WilliamsXX24} \\
        \cline{3-4}
         & & $\Omega(n^{1/3})$ & \textbf{Our Paper}$^+$ \\
        \cline{2-4}
         & \multirowcell{2}{$\cO(m)$} & $\Omega(n^{2/9})$ & \cite{HoppenworthXX25} \\
        \cline{3-4}
         & & $\Omega(n^{1/4})$ & \textbf{Our Paper}$^+$ \\
        \hline
        \multirowcell{5}{Hopset} & \multirowcell{2}{$\cO(n)$} & $\tilde{\Omega}(n^{1/2})$ & \cite{BodwinH23}$^*$ \\
        \cline{3-4}
         & & $\Omega(n^{1/2})$ & \textbf{Our Paper}$^+$ \\
        \cline{2-4}
         & \multirowcell{3}{$\cO(m)$} & $\tilde{\Omega}(n^{1/3})$ & \cite{BodwinH23}$^*$ \\
        \cline{3-4}
         & & $\Omega(n^{2/7})$ & \cite{HoppenworthXX25} \\
        \cline{3-4}
         & & $\Omega(n^{1/3})$ & \textbf{Our Paper}$^+$ \\
        \hline
    \end{tabular}
    \caption{Summary of our new results and best prior lower bounds ($^*$: weighted; $^+$: certified).}
    \label{tab:lb}
\end{table}

\begin{table}[tp]
    \centering
    \begin{tabular}{c|c|c|c}
        \hline
         & $|H|$ & $D$ & \\
        \hline
        \multirowcell{3}{Shortcut} & \multirowcell{4}{$\tilde{\cO}(n)$} & $n^{1/2+o(1)}$ & \cite{JambulapatiLS19}$^{+\diamond}$ \\
        \cline{3-4}
         & & $\cO(n^{1/2})$ & \cite{KoganP23,bals2025greedy}$^{+}$ \\
        \cline{3-4}
         & & $\cO(n^{1/3})$ & \cite{KoganP22,bals2025greedy}$^\dagger$ \\
        \cline{1-1}\cline{3-4}
        Hopset & & $\cO(n^{1/2})$ & \cite{UllmanY91,BermanRR10}$^\dagger$ \\
        \hline
    \end{tabular}
    \caption{Summary of best prior upper bounds ($^+$: certified; $^\dagger$: existential; $^\diamond$: parallel).}
    \label{tab:ub}
\end{table}

\subsection{Discussion}
\label{sec:discuss}

\paragraph{Comparison of constructions.}

We note that proofs of our improved diameter lower bounds for certified shortcuts can be concisely derived in a black-box way from the earliest construction, introduced in \cite{Hesse03} and later refined by \cite{HuangP21}, based on large convex sets, together with an improved alternation product due to \cite{LuWWX22}.

At a very high level, most known lower bound constructions for shortcuts are layered graphs where each layer consists of integral grid points in constant (either $2$ or $3$ to be exact) dimensions\footnote{Higher dimensions are studied in only a few prior works (e.g., \cite{Hesse03,WilliamsXX24}) when considering other related settings such as $\cO(n^{2-\epsilon})$-size shortcuts.}.
The general goal is to pack as many long paths as possible into the graph such that each path is the unique path between its endpoints and does not overlap too much with other paths\footnote{Some known constructions (e.g. \cite{WilliamsXX24,HoppenworthXX25}) allow long overlaps between a small number of critical paths. Nevertheless, such approach does not change the main ideas discussed here.}.
Each such path is called a critical path and is defined by its starting vertex and a critical direction, which specifies the exact edge to traverse in each layer.

These constructions are typically parameterized by two numbers $D$ and $r$.
The depth $D$ corresponds to the number of layers in the graph, which is equivalent to the number of edges in each critical path and thus should be maximized.
Meanwhile, $r$ typically controls the range of the critical directions.
In general, the goal is to pack a superlinear number of critical paths so that no linear-size shortcut can reduce the length of every critical path.
Due to the layered nature of the graphs, the number of starting vertices is smaller than the total number of vertices by a factor of $D$.
This has to be compensated by the number of critical directions, which is polynomial in $r$.
As a result, it demands that $r$ has a polynomial dependency on $D$.
To some extent, all later improvements can be viewed as optimizing this relation between $r$ and $D$ via more complicated constructions.

The most notable benefit of working with certified shortcuts is that such optimization is no longer required.
Indeed, in order to reduce the length of a critical path from $D$ to $D/2$, any certified shortcut has to add at least $D/2$ new edges.
Consequently, we gain a factor of $\Theta(D)$ for free, and it suffices to set $r$ to be any large enough constant.

It also seems plausible that our lower bounds are the best possible with current techniques.
Each layer needs to be at least $2$-dimensional to allow for any nontrivial set of critical directions, and an additional dimension is required for layering.
This suggests the demand of significantly new techniques to obtain diameter lower bounds better than $\Omega(n^{1/3})$ for $\cO(n)$-size shortcuts.
Regarding $\cO(m)$-size shortcuts, all known constructions crucially rely on alternation products to sparsify the graph, which unavoidably bring in another dimension.
As a result, $\Omega(n^{1/4})$ also appears to be a tough barrier to overcome with current techniques.
In fact, without the notion of certifiability, even achieving these lower bounds of $\Omega(n^{1/4})$ and $\Omega(n^{1/3})$ for the two settings seems out of reach due to $r$'s polynomial dependency on $D$.

\paragraph{Parallel constructions.}

We have seen that affirmative results on the existence of certified shortcuts could be valuable guidance in algorithm design, as \cref{lem:equi-seq-sc} means that any (almost-)linear-size certified shortcut has an efficient implementation (although only a nondeterministic one).
Recall that one major application of shortcuts is in parallel directed reachability.
One may wonder if the requirement of parallel constructions could impose more restrictive structural constraints and potentially lead to even stronger impossibility results.
This turns out not to be the case.
In \cref{sec:low-depth-certificates}, we show that certified shortcuts can always be constructed by parallel $2$-hop certified shortcutting procedures in logarithmic depth, while suffering only a logarithmic factor increase in the shortcut size.

%% file: prelim.tex
\section{Preliminaries}
\label{sec:prelim}

\paragraph{Notation.}

For $n \in \BN$, we use $[n]$ to denote the set of integers $\set{1,\dots,n}$.
$\langle \cdot, \cdot \rangle$ denotes inner products.
Bold letters are reserved for vectors.

For $r \ge 0$, let $\cB(r)$ be the set of all $2$-dimensional integral points within Euclidean distance $r$ to the origin, and $\cV(r)$ be the set of all vertices on the convex hull of $\cB(r)$ and with positive coordinates.
We use the following classical result by \cite{barany1998convex}.

\begin{lemma}[\cite{barany1998convex}]
\label{lem:ball-size}
It holds that $|\cV(r)|=\Theta(r^{2/3})$.
\end{lemma}

\paragraph{Graphs.}

Given a (possibly weighted) graph $G=(V,E)$, $n=|V|$ and $m=|E|$ are the number of vertices and edges respectively.
The \emph{transitive closure} of $G$ is the set of vertex pairs $(s,t)$ where $s$ can reach $t$.
For any path $P$, the \emph{length} of $P$ is the total weight of its edges, and we write $|P|$ for the total number of its edges.
For any two vertices $s,t \in V$, the \emph{distance} from $s$ to $t$, denoted $\dist_G(s,t)$, is the minimum length of any $s$-$t$ path.
We write $\dist_G(s,t)=\infty$ if $s$ cannot reach $t$.
The subscript $G$ is omitted when it is clear from context.

A \emph{chain} of a graph is a sequence of vertices $(v_1,v_2,\dots)$ such that $v_i$ can reach $v_{i+1}$ for all $i \ge 1$.
An \emph{$\ell$-chain cover} is a collection of at most $\ell$ vertex-disjoint chains such that any path contains at most $\cO(n/\ell)$ vertices that are not in any chain.

\paragraph{Shortcuts and hopsets.}

A \emph{shortcut} with \emph{diameter} $D$, or $D$-shortcut for short, is a set $H$ of additional (unweighted) edges, such that for any two vertices $s,t \in V$,
\begin{enumerate}
    \item If $t$ is not reachable from $s$ in $G$, then $t$ remains unreachable from $s$ in $G \cup H$;
    \item If $t$ is reachable from $s$ in $G$, then $t$ can be reached from $s$ using at most $D$ edges in $G \cup H$.
\end{enumerate}

Similarly, an (exact) \emph{hopset} with \emph{hopbound} $D$, or $D$-hopset, is a set $H$ of additional weighted edges, such that for any two vertices $s,t \in V$,
\begin{enumerate}
    \item It holds that $\dist_{G \cup H}(s,t) = \dist_G(s,t)$.
    \item If $\dist_G(s,t) < \infty$, then there exists an $s$-$t$ path in $G \cup H$ that uses at most $D$ edges and has length $\dist_G(s,t)$.
\end{enumerate}

Throughout, we often use the notion of \emph{critical paths}.
Critical paths for shortcuts are a set $\cP$ of paths such that each path $P \in \cP$ is the unique path between its endpoints.
The two endpoints of each path are called a \emph{critical pair}.
For hopsets, critical paths are only required to be the unique shortest path between their endpoints.

\paragraph{Certification.}

Given any shortcut $H$ for graph $G=(V,E)$, we say it is \emph{certified} if for any edge $(u,v) \in H$, there exist two edges $(u,w),(w,v) \in E \cup H$ for some vertex $w \in V$.
A hopset is said to be certified if it additionally satisfies that the weight of $(u,v)$ is equal to the total weight of $(u,w),(w,v)$.

\revised{%
\begin{lemma}
    \label{lem:intree-outtree-certified}

    Let $G$ be a graph, $T$ be an in-tree rooted at a vertex $v$, and $H$ the shortcut containing all edges $(u, v)$ for $u \in T$. Then, $H$ is certified.
    
    Symmetrically, let $G$ be a graph, $T$ be an out-tree rooted at a vertex $v$, and $H$ the shortcut containing all edges $(v, u)$ for $u \in T$. Then, $H$ is certified.
\end{lemma}
\begin{proof}
    For the in-tree case, pick some $(u, v) \in H$, and let $p$ be the parent of $u$ in $T$. Then, $G$ contains the edge $(u, p)$, and $p \in T$, thus $H$ contains the edge $(p, v)$. These two edges together certify $(u, v)$. The out-tree case is symmetric.
\end{proof}

\begin{lemma}
    \label{lem:certified-shortcut-subpaths}

    Let $P = (u_1, \dots, u_l)$ be a path. Then, there is a certified shortcut $H$ of size $O(l \log l)$, such that for any $1 \leq i < j \leq l$, either $(u_i, u_j) \in P$ or there is some $i < k < j$ such that $H \cup P$ contains the edges $(u_i, u_k)$ and $(u_k, u_j)$. 
\end{lemma}

\begin{proof}
    We will use a standard divide-and-conquer approach, showing that a shortcut $H$ of size $l \lceil \log l \rceil$ exists through induction on $l$. The base case $l = 1$ is trivial. 
    
    Suppose $l > 1$, and let $k = \lfloor (l + 1) / 2 \rfloor$ be the middle point of the path. Let $H_{< k}$ be the (possibly empty) shortcut constructed by applying the induction assumption to the (possibly empty) subpath $(u_1, \dots, u_{k - 1})$, $H_{> k}$ the corresponding shortcut constructed by applying the induction assumption to the subpath $(u_{k + 1}, \dots, u_l)$, and let $H_k$ be the shortcut containing all edges $(u_i, u_k)$ for $i < k - 1$ and $(u_k, u_j)$ for $j > k + 1$.

    We let $H = H_{< k} \cup H_k \cup H_{> k}$. This shortcut is certified, since all three individual shortcuts are certified, $H_{< k}$ and $H_{> k}$ by induction, and $H_k$ by \Cref{lem:intree-outtree-certified}. The shortcut has size at most $(l - 3) + (l - 1) (\lceil \log l \rceil - 1) \leq l \lceil \log l \rceil$.
    
    It remains to check the desired property. If $i, j < k$, $H_{< k} \cup P$ contains the desired edges by induction, and the case $i, j > k$ is symmetric with $H_{> k}$. If $i < k < j$, the property holds by construction of $H_k$.
\end{proof}

\lemtwosteps*

\begin{proof}
    It suffices to consider just one path $P = (u_1, \dots, u_l)$, in which case $t = l + |H|$. By \Cref{lem:equi-seq-sc}, it suffices to show that a certified shortcut $H' \supseteq H$ of size $\tilde{O}(l + |H|)$ exists.
    
    Let $H''$ be the certified shortcut constructed by \Cref{lem:certified-shortcut-subpaths} on $P$. We let $H' = H'' \cup H$. The shortcut has size $O(l \log l + |H|)$ as desired, and it is certified, since for any $(u_i, u_j) \in H$, $H'' \cup P$ either contains $(u_i, u_j)$ or the contains the edges $(u_i, u_k)$ and $(u_k, u_j)$ for some $k$.
\end{proof}

We remark that the $O(\log n)$ factor in \Cref{lem:certified-shortcut-subpaths} and \Cref{lem:lem-two-steps} can be improved to $O(\log^* n)$ or even further using \cite{Raskhodnikova10}.
}

As a measure of how close it is to being certified, the \emph{certification complexity} of a shortcut (resp. hopset) $H$ is defined to be the smallest size of any certified shortcut (resp. hopset) $H' \supseteq H$.

\begin{lemma}
\label{lem:shortcut-size}
Let $G$ be a DAG and $\cP$ be a set of critical paths for shortcuts, each with $D$ edges, such that any two paths intersect on at most a single edge.
Any certified $D/2$-shortcut $H$ has size at least $|\cP| \cdot D/2$.
\end{lemma}

\begin{proof}
Fix a path $P \in \cP$ with endpoints $s,t$.
Since $H$ is an $D/2$-shortcut, there exists an $s$-$t$ path $P'$ in $G \cup H$ with $|P'| \le D/2$.
Consider the following iterative procedure starting with $i=0$ and $P_0=P'$.
Whenever $P_i$ contains an edge $e=(u,v) \not\in E$, obtain $P_{i+1}$ from $P_i$ by replacing $e$ with its two certifying edges $e'=(u,w)$ and $e''=(w,v)$, and continue the procedure by increasing $i$.
Note that adding $H$ does not change the transitive closure of $G$ and that $P$ is the unique $s$-$t$ path in $G$, so all $P_i$ can only contain vertices on $P$.
As $G$ is a DAG and $|P_i|$ increases by $1$ in each iteration, the above procedure must stop at iteration $i^*=|P|-|P'| \ge D/2$.
This means $H$ contains at least $D/2$ new edges with both their endpoints on $P$.
Moreover, all these edges are distinct for each critical path because any two critical paths intersect on at most a single edge.
This concludes the proof.
\end{proof}

The above argument can be extended to hopsets.
We omit the proof as it is almost identical.

\begin{lemma}
\label{lem:hopset-size}
Let $G$ be an undirected graph and $\cP$ be a set of critical paths for hopsets, each with $D$ edges, such that any two paths intersect on at most a single edge.
Any certified $D/2$-hopset $H$ has size at least $|\cP| \cdot D/2$.
\end{lemma}

\paragraph{Prior constructions.}

Our lower bounds will build upon the following constructions. 

\begin{lemma}[\cite{HuangP21}]
\label{lem:hp}
For $d,D,r > 0$, there exists a family of $D$-layered DAGs such that:
\begin{enumerate}
    \item \label{item:hp-N} It has $n=\Theta(D^{d+1} r^d)$ vertices and $m=\Theta(n\Delta)$ edges, where $\Delta=\Theta(r^{d(d-1)/(d+1)})$ is the maximum degree.
    \item \label{item:hp-P} It has a set $\cP$ of $\Theta(n\Delta/D)$ critical paths such that:
        \begin{enumerate}
            \item \label{item:hp-P-length} Each path is from the first layer to the last layer.
            \item \label{item:hp-P-unique} Each path is the unique path between its endpoints.
            \item \label{item:hp-P-overlap} Any two paths are edge-disjoint.
        \end{enumerate}
\end{enumerate}
\end{lemma}

\begin{lemma}[\cite{LuWWX22}]
\label{lem:lvwx}
For $d,D,r > 0$, there exists a family of $dD$-layered DAGs such that:
\begin{enumerate}
    \item \label{item:lvwx-N} It has $n=\Theta(d D^{d+2} r^{d+1})$ vertices and $m=\Theta(n\Delta)$ edges, where $\Delta=\Theta(r^{2/3})$ is the maximum degree.
    \item \label{item:lvwx-P} It has a set $\cP$ of $\Theta(n\Delta^d/(dD))$ critical paths such that:
        \begin{enumerate}
            \item \label{item:lvwx-P-length} Each path is from the first layer to the last layer.
            \item \label{item:lvwx-P-unique} Each path is the unique path between its endpoints.
            \item \label{item:lvwx-P-overlap} The intersection between any two paths is a subpath with at most $d-1$ edges.
        \end{enumerate}
\end{enumerate}
\end{lemma}

%% file: lb.tex
\section{Lower Bounds for Certified Shortcuts and Hopsets}
\label{sec:lb}

\subsection{Shortcuts}

In this section, we obtain lower bounds for certified shortcuts, hence proving \cref{thm:lb-sc-m,thm:lb-sc-n}.

\revised{%
\begin{lemma}\label{lem:edge-size-lowerbound}
Let $C>0$ be a sufficiently small universal constant.
For any $\epsilon \in [0,2/9]$, there exists a family of DAGs $G$ with $n$ vertices and $m=\Theta(n^{1+\epsilon})$ edges, such that any certified shortcut $H$ of size $m n^\epsilon$ has diameter at least $C \cdot n^{1/4-9\epsilon/8}$.
\end{lemma}
}

\revised{%
\begin{proof}
Let $G$ be the graph promised by \cref{lem:lvwx} for $d=2$ and $D=\Theta(n^{1/4-9\epsilon/8})$.
Thus we have $r=\Theta(n^{3\epsilon/2})$ and $\Delta=\Theta(n^\epsilon)$.
By \cref{item:lvwx-N,item:lvwx-P}, we have $m = \Theta(n^{1+\epsilon})$ edges and $|\cP| = \Theta(m n^\epsilon / D)$ critical paths, each of which has $2D$ edges and is the unique shortest path between its endpoints.
Additionally, any two of these critical paths share at most one edge.
By \cref{lem:shortcut-size}, any certified shortcut $H$ of size $|\cP| \cdot D = \cO(m n^\epsilon)$ has diameter at least $D$.
Adjusting $D$ by a constant factor concludes the proof.
\end{proof}
}

\revised{%
\begin{lemma}\label{lem:node-size-lowerbound}
Let $C>0$ be a sufficiently small universal constant.
For any $\epsilon \in [0,1/3]$, there exists a family of DAGs $G$ with $n$ vertices, such that any certified shortcut $H$ of size $n^{1+\epsilon}$ has diameter at least $C \cdot n^{1/3 - \epsilon}$.
\end{lemma}
}

\revised{%
\begin{proof}
Let $G$ be the graph promised by \cref{lem:hp} for $d=2$ and $D=\Theta(n^{1/3-\epsilon})$.
Thus we have $r=\Theta(n^{3\epsilon/2})$ and $\Delta=\Theta(n^\epsilon)$.
By \cref{item:hp-P}, there are $|\cP| = \Theta(n^{1+\epsilon}/D)$ critical paths, each of which has $D$ edges and is the unique shortest path between its endpoints.
Additionally, all these critical paths are edge disjoint.
By \cref{lem:shortcut-size}, any certified shortcut $H$ of size $|\cP| \cdot D/2 = \cO(n^{1+\epsilon}) $ has diameter at least $D/2$.
Adjusting $D$ by a constant factor concludes the proof.
\end{proof}
}

\subsection{Hopsets}

In this section, we obtain lower bounds for certified hopsets, hence proving \cref{thm:lb-hs-m,thm:lb-hs-n}.

\revised{%
\begin{lemma}
\label{lem:lb-hs-n}
Let $C>0$ be a sufficiently small universal constant.
For any $\epsilon \in [0,1/3]$, there exists a family of unweighted undirected graphs $G$ with $n$ vertices, such that any certified hopset $H$ of size $n^{1+\epsilon}$ has hopbound at least $C \cdot n^{1/2-3\epsilon/2}$.
\end{lemma}
}

\revised{%
\begin{lemma}
\label{lem:lb-hs-m}
Let $C>0$ be a sufficiently small universal constant.
For any $\epsilon \in [0,2/9]$, there exists a family of unweighted undirected graphs $G$ with $n$ vertices and $m=\Theta(n^{1+\epsilon})$ edges, such that any certified hopset $H$ of size $m n^\epsilon$ has hopbound at least $C \cdot n^{1/3-3\epsilon/2}$.
\end{lemma}
}

The proofs of \cref{lem:lb-hs-n,lem:lb-hs-m} are based on graph constructions in \cref{lem:hs}, whose proof is deferred to \cref{sec:proof-hs}.

\begin{restatable}[]{lemma}{lemhs}
\label{lem:hs}
For $d \in \set{1,2}$ and $D,r>0$, there exists a family of unweighted undirected graphs\footnote{The requirement $d \in \set{1,2}$ is due to technical reasons in undirected settings. If we are only concerned with DAGs, the same lemma holds as is for any $d>0$. See \cref{sec:proof-hs} for more details.} such that:
\begin{enumerate}
    \item \label{item:hs-N} It has $n=\Theta(dD^{d+1}r^{d+1})$ vertices and $m=\Theta(n\Delta)$ edges, where $\Delta=\Theta(r^{2/3})$ is the maximum degree.
    \item \label{item:hs-P} It has a set $\cP$ of $\Theta(n\Delta^d/(dD))$ critical paths such that:
        \begin{enumerate}
            \item \label{item:hs-P-length} Each path has $dD$ edges.
            \item \label{item:hs-P-unique} Each path is the unique shortest path between its endpoints.
            \item \label{item:hs-P-overlap} The intersection between any two paths is a subpath with at most $d-1$ edges.
        \end{enumerate}
\end{enumerate}
\end{restatable}

We now use the graph construction of \cref{lem:hs} to prove the desired lower bounds.

\begin{proof}[Proof of \cref{lem:lb-hs-n}]
Let $G$ be the graph promised by \cref{lem:hs} for $d=1$ and $D=\Theta(n^{1/2-3\epsilon/2})$.
Thus we have $r=\Theta(n^{3\epsilon/2})$ and $\Delta=\Theta(n^\epsilon)$.
By \cref{item:hs-P}, there are $|\cP| = \Theta(n^{1+\epsilon}/D)$ critical paths, each of which has $D$ edges and is the unique shortest path between its endpoints.
Additionally, all these critical paths are edge disjoint.
By \cref{lem:hopset-size}, any certified hopset $H$ of size $|\cP| \cdot D/2 = \cO(n^{1+\epsilon}) $ has hopbound at least $D/2$.
Adjusting $D$ by a constant factor concludes the proof.
\end{proof}

\begin{proof}[Proof of \cref{lem:lb-hs-m}]
Let $G$ be the graph promised by \cref{lem:hs} for $d=2$ and $D=\Theta(n^{1/3-3\epsilon/2})$.
Thus we have $r=\Theta(n^{3\epsilon/2})$ and $\Delta=\Theta(n^\epsilon)$.
By \cref{item:hs-N,item:hs-P}, we have $m = \Theta(n^{1+\epsilon})$ edges and $|\cP| = \Theta(m n^\epsilon / D)$ critical paths, each of which has $2D$ edges and is the unique shortest path between its endpoints.
Additionally, any two of these critical paths share at most one edge.
By \cref{lem:hopset-size}, any certified hopset $H$ of size $|\cP| \cdot D = \cO(m n^\epsilon)$ has hopbound at least $D$.
Adjusting $D$ by a constant factor concludes the proof.
\end{proof}

%% file: existing.tex
\section{Certifiability and Non-Certifiability of Existing Approaches}
\label{sec:existing}

In this section, we give further evidence showing that our definition of certifiability indeed captures the notion of constructiveness, instead of just existence, of shortcuts and hopsets, as claimed in \cref{thm:existing}.
Specifically, we observe in \cref{sec:existing-cert} that shortcuts constructed by known efficient algorithms are either certified or at least have low certification complexity.
On the other hand, in \cref{sec:existing-noncert}, we prove unconditional lower bounds suggesting that known existential proofs may not have efficient implementations to construct certified shortcuts or hopsets due to large certification complexity.

In the proofs, we will very briefly summarize the high-level ideas of each approach that are relevant.
For more details, we refer readers to the corresponding papers.

\subsection{Known Efficient Algorithms are Certifiable}
\label{sec:existing-cert}

As discussed in the prior work, there are two known approaches to designing efficient algorithms for constructing shortcuts:
\begin{itemize}
    \item The pivot-based approach of \cite{Fineman18} obtaining diameter $\tilde{O}(n^{2/3})$, later refined by \cite{JambulapatiLS19} obtaining diameter $n^{1/2 + o(1)}$. Notably, these are parallel algorithms, with work $\tilde{O}(m)$ and respective depths
    matching their diameter bounds.
    \item The chain cover-based approach of \cite{KoganP22}, with two different, almost-linear time algorithms from \cite{KoganP23} and \cite{bals2025greedy} (Section 6). These algorithms are sequential, but are deterministic and obtain diameter $O(\sqrt{n})$.
\end{itemize}

In \Cref{sec:pivot-alg}, we will observe that the shortcuts produced by the pivot-based algorithms are in fact always certified. This follows immediately from \cref{lem:intree-outtree-certified}, as the algorithms only add edges in the form of shortcuts for subtrees, connecting all vertices in the tree to the root.

\begin{restatable*}{observation}{pivotbasedcertified} \label{obs:pivot-based-certified}
    The shortcuts produced by both the sequential and parallel versions of the algorithms of \cite{Fineman18} and \cite{JambulapatiLS19} are certified, and thus in particular have certification complexity equal to their size of $\tilde{O}(n)$.
\end{restatable*}

In \Cref{sec:chain-alg}, we consider the two algorithms of \cite{KoganP23} and \cite{bals2025greedy}. From the former algorithm's analysis, one can easily observe the existence of a certified extension of asymptotically the same size as the shortcut. For the latter algorithm, bounding the certification complexity is more involved, requiring a closer look into the data structures it uses, specifically balanced binary trees. We will bound the certification complexity by the work done to maintain these binary trees, and thus only obtain a certification complexity of $\tilde{O}(m)$, rather than $\tilde{O}(n)$.

\begin{restatable*}{observation}{diamdomibasedcertified} \label{obs:diam-domi-based-certified}
    The $O(\sqrt{n})$-diameter, $\tilde{O}(n)$-size shortcut produced by the almost-linear-time algorithm of \cite{KoganP23} has certification complexity $\tilde{O}(n)$.
\end{restatable*}
\begin{restatable*}{lemma}{coverbasedcertified} \label{lem:cover-based-certified}
    The $O(\sqrt{n})$-diameter, $\tilde{O}(n)$-size shortcut produced by the almost-linear-time algorithm of \cite{bals2025greedy} has certification complexity $\tilde{O}(m)$.
\end{restatable*}

\subsubsection{Pivot-Based Algorithms}\label{sec:pivot-alg}

In the sequential version of the pivot-based approach of \cite{Fineman18} (given in \Cref{alg:fineman}), a random \textit{pivot} vertex $x$ is selected. Then, shortcuts are added from every vertex that can reach $x$ to $x$, and from $x$ to every vertex that can be reached from $x$. With these shortcuts in place, the algorithm then recursively constructs shortcuts for three disjoint subsets of vertices: those that cannot reach $x$ or be reached from $x$, those that can only reach $x$, and those that can only be reached from $x$.

Note that it is highly nontrivial that adding shortcuts in this way reduces the lengths of shortest paths between vertices that appear in different subsets. We do not replicate the analysis from Fineman's paper \cite{Fineman18}, as a full understanding of the algorithm is not necessary to observe that any shortcut constructed by the algorithm is already certified by \Cref{lem:intree-outtree-certified}.

\begin{algorithm}[tp]
    \caption{The sequential version of Fineman's shortcutting algorithm \cite{Fineman18}.}
    \label{alg:fineman}
    \begin{algorithmic}[1]
        \Function{Fineman}{$G=(V, E)$}
            \If{$V = \emptyset$} \Return $\emptyset$ \EndIf
            \State Select a pivot $x \in V$ uniformly at random
            \State Let $R^+$ denote the set of vertices reachable from $x$ in $G$
            \State Let $R^-$ denote the set of vertices that can reach $x$ in $G$
            \State \tikz[baseline=(a.base),inner sep=0.01cm]\node[fill=black!15](a){
                Let $S := \{(x, v) \mid v \in R^+\}\ \cup\ \{(v, x) \mid v \in R^-\}$
            };
            \State $V_B := R^+ \cap R^-$
            \State $V_S := R^+ \setminus V_B$
            \State $V_P := R^- \setminus V_B$
            \State $V_R := V \setminus (V_B\ \cup\ V_S\ \cup\ V_P)$ 
            \State \Return $S\ \cup\ \Call{Fineman}{G[V_B]}\ \cup\ \Call{Fineman}{G[V_P]}\ \cup\ \Call{Fineman}{G[V_R]}$
        \EndFunction
    \end{algorithmic}
\end{algorithm}

The sequential version of the shortcut algorithm of \cite{JambulapatiLS19} is given in \Cref{alg:janbulapati}. The core idea behind their approach is to sample multiple pivots at once, add shortcut edges to and from each of the sampled pivots, and then to recurse into disjoint induced subgraphs determined by reachability from the whole pivot set. The analysis of the algorithm, including why sampling multiple pivots at once results in better diameter guarantees, is once again quite complicated. Thankfully, we can easily observe that the changes to the algorithm do not change that the shortcut produced is certified, by exactly the same argument as for \cite{Fineman18}.

For the parallel versions of the two algorithms, the sets $R^+$ and $R^-$ are limited to only contain vertices that can reach the pivot or be reached from the pivot by paths of at most a certain length (as $R^+$ and $R^-$ can then be constructed with depth proportional to this distance limit). This change requires the rest of the algorithm to get significantly more complex, but again does not change that the produced shortcut is certified: if a vertex is contained in $R^-$, so is every shortest path from it to the pivot.

As claimed prior, we can thus observe the following:

\pivotbasedcertified

\begin{algorithm}[tp]
    \caption{The sequential version of Jambulapati-Liu-Sidford's shortcutting algorithm \cite{JambulapatiLS19}.}
    \label{alg:janbulapati}
    \begin{algorithmic}[1]
        \Function{JLS}{$G=(V, E)$, $k$, $r$}
            \State Let $p_r := \min \left\{1, \frac{20 k^{r + 1} \log n}{n}\right\}$
            \State Let $X$ be a subset of vertices containing each vertex independently with probability $p_r$
            \State $S \gets \emptyset$
            \For{$x \in X$}
                \State Let $R^+$ denote the set of vertices reachable from $x$ in $G$
                \State Let $R^-$ denote the set of vertices that can reach $x$ in $G$
                \State \tikz[baseline=(a.base),inner sep=0.01cm]\node[fill=black!15](a){
                    $S \gets S \cup \{(x, v) \mid v \in R^+\}\ \cup\ \{(v, x) \mid v \in R^-\}$
                };
                \State Add label $(x, +)$ to all vertices of $R^+ \setminus R^-$
                \State Add label $(x, -)$ to all vertices of $R^- \setminus R^+$
                \State Add label \textsc{DONE} to all vertices of $R^+ \cap R^-$
            \EndFor
            \State Let $W$ be the subset of vertices that do not have the label \textsc{DONE}
            \State Let $V_1, \dots, V_\ell$ be a partition of $W$ by unique label combination
            \For{$i \in [\ell]$}
                \State $S \gets S \cup \Call{JLS}{G[V_i], k, r + 1}$
            \EndFor
            \State \Return $S$
        \EndFunction
    \end{algorithmic}
\end{algorithm}

\subsubsection{Chain Cover-Based Algorithms} \label{sec:chain-alg}

Recall from the preliminaries the definitions of a \textit{chain} and an \textit{$\ell$-chain cover}: on a graph $G$, a chain $C = (v_1, \dots, v_k)$ is a sequence of vertices where $v_i$ can reach $v_{i + 1}$ for each $i \in [k - 1]$, and an $\ell$-chain cover is a collection of at most $\ell$ vertex-disjoint chains, such that any path on $G$ contains at most $O(n / \ell)$ vertices that are not part of any chain. By \cite{Raskhodnikova10}, we can always add $\tilde{O}(k)$ edges to reduce the diameter of any chain $C$ to $2$.
Meanwhile, it is not hard to verify that the edges produced by the construction of \cite{Raskhodnikova10} are certified, given that consecutive vertices on the chain are already connected by edges.

Let $C_1, \dots, C_\ell$ be an $\ell$-chain cover, and let $H_1$ be a shortcut for each of the chains. Then, in the graph $G \cup H_1$, any shortest path contains at most three vertices from each chain: if it contained more, one could replace the subpath from the first to the last with a path of length $2$, reducing the length of the shortest path, a contradiction. Additionally, any simple path (including any shortest path) contains at most $O(n / \ell)$ vertices that are not part of any chain. Thus, the diameter of the graph $G \cup H_1$ is at most $O(\ell + n / \ell)$, which in particular is $O(\sqrt{n})$ for $\ell = \Theta(\sqrt{n})$.

The $O(n^{1/3})$-diameter shortcut $H = H_1 \cup H_2$ of \cite{KoganP22} includes a second part $H_2$, with shortcut edges from randomly sampled vertices to a randomly sampled subset of chains. Unfortunately, computing the shortcut $H_2$ is conditionally hard in almost-linear time \cite{KoganP23}.\footnote{To give further evidence supporting this hardness, we show later in \Cref{sec:existing-noncert-sample} that the shortcut $H_2$ indeed has certification complexity $m^{1 + \Omega(1)}$. See \cref{lem:noncert-kp} for more details.}
As a result, known almost-linear-time algorithms cannot deal with this part, and consequently only obtain shortcuts of diameter $O(\sqrt{n})$. 

\paragraph{Computing a chain cover.} As discussed above, to construct an $O(\sqrt{n})$-diameter shortcut, it suffices to compute an $\ell$-chain cover for $\ell=O(\sqrt{n})$. \cite{KoganP22fast} gives a method to extract such a chain cover from a min-cost $\ell$-value flow on a gadget graph $\tilde{G}$, constructed as follows:
\begin{itemize}
    \item Add a source $s$ and a sink $t$.
    \item For every vertex $v \in G$,
    \begin{itemize}
        \item Add two vertices $v^{in}$ and $v^{out}$.
        \item Add two parallel edges from $v^{in}$ to $v^{out}$, one with capacity $1$ and cost $-1$, and the other with capacity $\infty$ and cost $0$.
        \item Add edges from $s$ to $v^{in}$ and $v^{out}$ to $t$ with capacity $\infty$ and cost $0$.
    \end{itemize}
    \item Finally, for every edge $e = (v, u) \in E$, add an edge from $v^{out}$ to $u^{in}$ of capacity $\infty$ and cost $0$.
\end{itemize}
Let $P_1, \dots, P_\ell$ be an arbitrary integral flow decomposition of a min-cost $\ell$-value flow on $\tilde{G}$. For each path $P_i$, we can define a corresponding chain $C_i$, by including in the chain the vertices $v$ for which the path $P_i$ takes the unit-capacity edge from $v^{in}$ to $v^{out}$ with cost $-1$. These chains form an $\ell$-chain cover: suppose that there existed some path $P'$ in the graph containing strictly more than $n / \ell$ vertices not in any chain. Then, the min-cost flow would not be optimal: the cost contribution of the worst path $P_i$ is at least $-n / \ell$, as the total cost of all paths is at least $-n$, while adding $P'$ to the flow would cost strictly less than $-n / \ell$.

The graph $\tilde{G}$ has asymptotically equal size with $G$, thus minimum cost flow can be solved exactly on it in time $m^{1 + o(1)}$ \cite{Brand0PKLGSS23flow}.
However, extracting the chains $C_1, \dots, C_\ell$ naively takes time linear in the total length of the paths $P_1, \dots, P_\ell$, which could be up to $n$ times the length of the longest path in the graph. The algorithms of \cite{KoganP23} and \cite{bals2025greedy} avoid this problem in two distinct ways.

\paragraph{The approach of \cite{KoganP23}.} \cite{KoganP23} introduces a weaker version of an $\ell$-chain cover that they call a \textit{diameter-dominating path set}. A set of $\ell$ chains $C_1, \dots, C_\ell$ is a diameter-dominating path set of $G$, if every path on $G$ of length at most $D = \frac{4 n}{\ell}$ contains at most $D / 2$ vertices not in any chain.

While diameter-dominating path sets have a weaker covering property, shortcutting one is sufficient to reduce the graph's diameter by a constant fraction. Indeed, suppose that we select $\ell$ so that $G$ has diameter at most $D$, and let $H$ be a shortcut for a diameter-dominating path set $C_1, \dots, C_\ell$. Due to the diameter bound, any shortest path on $G$ has length at most $D$, and thus contains at most $D / 2$ vertices not on any of the chains $C_i$. As we can shorten this path by traversing through edges in $H$ so that the path contains at most three vertices from each chain, we obtain that the diameter of the graph $G \cup H$ is at most $\frac{1}{2} D + O(\ell) = \frac{1}{2} D + O(n / D)$, which is smaller than $D$ by a constant factor for large enough $D = \Omega(\sqrt{n})$. Thus, one can iteratively construct and shortcut a diameter-dominating path set $O(\log n)$ times to obtain an $O(\sqrt{n})$-shortcut.

The advantage of considering diameter-dominating path sets rather than chain covers is the ability to modify the gadget graph $\tilde{G}$ so that the total length of the paths $P_i$ in any flow decomposition of any optimal min-cost flow is small, in particular $O(n)$.

Indeed, to obtain a diameter-dominating path set, \cite{KoganP23} slightly changes the construction of the gadget graph $\tilde{G}$, replacing the edges between each pair $v^{in}$ and $v^{out}$ with an edge of capacity $1$ and cost $-2$, and an edge of capacity $\infty$ and \emph{positive} cost $1$. Since the total capacity-times-cost of negative-cost edges on this modified $\tilde{G}$ is $-2n$, at most $2n$ flow can go over the $1$-cost edges, thus in particular the total length of the paths $P_i$ in the flow decomposition must be at most $O(n)$. That $C_1, \dots, C_\ell$ form a diameter-dominating path set can be proven like the earlier proof that the chains in the unmodified $\tilde{G}$ form an $\ell$-cover, by contradiction on the optimality of the flow.

Given the analysis of the algorithm, it is now clear how to obtain a certified extension $H' \supset H$ of small size. We can simply let $H'$ be a shortcut of all of the paths $P_i$ which determine the chains $C_i$. Since the total length of these paths is $O(n)$, the size of this $H'$ is $\tilde{O}(n)$.

As claimed prior, we can thus observe the following:

\diamdomibasedcertified

\paragraph{The approach of \cite{bals2025greedy}.} \cite{bals2025greedy} observes that the technique of \cite{Caceres23} for constructing minimal chain covers can also be used to extract the chains $C_1, \dots, C_\ell$ directly from a flow on $\tilde{G}$ in $O(m \log n)$ time.

First, map the min-cost $\ell$-flow on $\tilde{G}$ onto a flow on $G' = G \cup \{s, t\}$ by contracting each pair $v^{in}, v^{out}$ onto the single vertex $v$. Let $P_1, \dots, P_\ell$ be the paths of some flow decomposition. Now, each vertex $v$ can be included in the chain $C_i$ of an arbitrary path $P_i$ that goes through the vertex.
The approach is to construct the paths $P_1, \dots, P_\ell$ implicitly, by maintaining the indices of paths that go through each vertex $v$ in a balanced binary tree data structure, supporting each of the following three operations efficiently (i.e. in time $O(\log n)$):
\begin{itemize}
    \item \textsc{join}, taking two trees and merging them into a single tree.
    \item \textsc{split}, taking a tree and a desired size $k$, and splitting the tree (of size $k'$) into two trees of sizes $k$ and $k' - k$.
    \item \textsc{member}, returning an arbitrary member of the tree.
\end{itemize}

Constructing the chains proceeds as follows: initially, create a tree containing all path indices $1, \dots, \ell$ for the source $s$. Then, iterate over vertices in topological order, starting from $s$ \footnote{When constructing shortcuts, it is without loss of generality to assume the graph is a DAG, as we can always shortcut rooted spanning trees in each strongly connected component and contract them. By \cref{lem:intree-outtree-certified}, the added edges during this process are also certified.}. For each vertex $v$ considered, suppose that by induction we have already constructed a tree containing the paths that visit $v$. We then do the following:
\begin{enumerate}
    \item If $v \not\in \{s, t\}$, assign the vertex to the chain $C_i$ of an arbitrary member $i$ of the tree, or to no chain if the tree is empty.
    \item Split the tree into trees of sizes $f_1, \dots, f_{\mathrm{deg}_{G'}^+(v)}$, where $f_1, \dots, f_{\mathrm{deg}_{G'}^+(v)}$ are the flow values on the outgoing edges from $v$.
    \item For each outgoing edge $(v, u)$, merge the corresponding tree (of size equal to the edge's flow value) produced in step 2. into the tree of the other endpoint $u$.
\end{enumerate}
Correctness is immediate, as a number of paths equal to the flow value go over each edge. Extracting chains $C_1, \dots, C_\ell$ this way takes $O(m \log n)$ time, as the number of tree operations required is $O(m)$.

We will now prove \Cref{lem:cover-based-certified}. For this proof, we will consider using a treap \cite{AragonS89} as the balanced binary tree data structure\footnote{Note that treaps are randomized, and using them would not lead to a deterministic shortcut algorithm. \cite{Caceres23} used a different, more complex deterministic structure instead. For simplicity of analysis, we consider treaps, as they allow for a simpler construction of the certified extension. It is possible to extend our analysis to other data structures used to support the algorithm.}. Since the proof uses the internal structure of treaps, we include a short overview of the data structure in \Cref{sec:treaps} for completeness.

\coverbasedcertified

\begin{proof}
    Suppose that \textsc{member} always returns the root of the treap, that the total work done over the treap-operations is $O(m \log n)$ (which holds with high probability), and that the sampled priorities are increasing in index, so that path index $1$ has the lowest priority, $2$ the second lowest, and so on (which can be assumed to hold without loss of generality).
    
    For an implicit path $P_i$, we define a vertex $v \in P_i$ to be \textit{important} if either
    \begin{enumerate}
        \item $i$ was the root vertex of the treap at $v$.
        \item Either (any vertex in) the subtree of $i$ or the parent of $i$ is different between the treap at $v$ and the treap at the vertex \emph{preceding} $v$ on $P_i$.
        \item Either (any vertex in) the subtree of $i$ or the parent of $i$ is different between the treap at $v$ and the treap at the vertex \emph{following} $v$ on $P_i$.
        \end{enumerate}
    We define the \textit{important chain} $C'_i$ of the path to be the chain containing its important vertices.
    
    First, notice that the total size of the important chains $C'_i$ is $O(m \log n)$: the total number of times vertices are important due to being the root vertex of a treap is exactly $n$, and as the treap operations take $O(m \log n)$ total work, and each operation is performed recursively down from the root, the total number of times that vertices see their subtrees change is $O(m \log n)$.

    We will construct a certified shortcut $H' = \bigcup_i H'_i$ of the chains $C'_i$ of size $\tilde{O}(m)$. Since $C_i \subseteq C'_i$ (as \textsc{member} returns the root, and a vertex is important for the root of the treap at that vertex), this is sufficient.

    Specifically, we will construct a certified shortcut $H'_i$ of size $\tilde{O}(|C'_i|)$ on $G \cup \bigcup_{j < i} H'_j$ for every $i$ in increasing order. First, for $i = 1$, since the chain $C'_i$ is a path in the graph (as the minimum-priority element is always the root of the treap it is in), we can use \Cref{lem:certified-shortcut-subpaths} to construct a certified shortcut of size $\tilde{O}(|C'_1|)$ for it that reduces the diameter of the chain to $2$. Now, suppose that $i > 1$. In this case, the chain $C'_i$ is not necessarily a path in the graph. We will first show that each of the edges $(u, v)$ between two consecutive vertices $u, v$ on the chain is in fact already certified in $G \cup \bigcup_{j < i} H'_j$ (i.e. that there exist edges $(u, p)$ and $(p, v)$ in this partially constructed shortcut graph). This suffices, as it reduces the situation back to the path case.

    Now, let $u$ and $v$ be consecutive vertices on the important chain $C'_i$. If they are consecutive vertices on $P_i$, we are done, as the edge $(u, v)$ exists in the graph. Otherwise, we know that $i$ must not be the root vertex of the treaps of either $u$ or $v$, and the parent vertex $j$ of $i$ is the same in both the treap at $u$ and the treap at $v$, because having different parent vertices implies that either $u,v$ are consecutive on $P_i$ or $u,v$ are not consecutive on $C'_i$. Since $j$ has smaller priority than $i$ as $i$'s parent, we have $j < i$. Additionally, it must be that $u$ and $v$ are also important vertices for $j$, since if the subtree or parent of $i$ changes, the subtree or parent of $i$'s parent $j$ must change as well. Thus, by induction, there exists a path of length $2$ from $u \in C'_j$ to $v \in C'_j$ in $G \cup \bigcup_{j' \le j} H'_{j'}$. The edges $(u, p)$ and $(p, v)$ of this path satisfy our requirement.
\end{proof}

\subsection{Known Existential Proofs are Non-Certifiable}
\label{sec:existing-noncert}

All existential proofs mainly fall into two categories.
The folklore argument \cite{UllmanY91} and its improvement \cite{KoganP22} heavily rely on sampling of vertices and paths, whereas \cite{BermanRR10,bals2025greedy} utilize greedy approaches.
In the rest of this section, we show that all of them construct shortcuts and hopsets with large certification complexity.
By the equivalence in \cref{lem:equi-seq-sc}, having superlinear certification complexity immediately implies that there are no efficient implementations (as certified shortcutting procedures) that can generate all edges of the almost-linear-size shortcuts or hopsets guaranteed by known existential proofs.

In addition, we show in \cref{sec:existing-noncert-cc} that the new chain cover-based approach of \cite{bals2025greedy} (Section 6) is highly sensitive to the choice of the chain cover extracted from the flow instance.
Not all chain covers are constructive.
More specifically, if the chain cover could be adversarially chosen, then even certifying the chain edges alone would lead to superlinear certification complexity.
On the contrary, we presented in \cref{sec:existing-cert} the existence of chain covers that can be efficiently extracted from the flow instance and have almost-linear certification complexity.
Altogether, it indicates that efficient implementations of such approach crucially utilizes certain nontrivial structures of the extracted chain covers.
We also emphasize that unlike sampling-based approaches, no conditional lower bound based on fine-grained complexity is known for these new chain cover-based algorithms.

\subsubsection{Sampling-Based Approaches}
\label{sec:existing-noncert-sample}

The critical and seemingly most time consuming step in implementing any sampling-based approach \cite{UllmanY91,KoganP22} is to check the reachability (or even to compute the shortest paths in the case of hopsets) between many pairs of vertices or vertex-path pairs.
In fact, \cite{KoganP23} show that such approach requires substantially superlinear time under standard computational assumptions in fine-grained complexity.

Combining the notion of certifiability with some ideas of \cite{KoganP23}, we are actually able to prove unconditional lower bounds on the certification complexity of these sampling-based approaches. 

\begin{lemma}
\label{lem:noncert-uy-sc}
For $\epsilon \in [0,1/25)$, there exists a family of DAGs with $n$ vertices and $m=\Theta(n^{1+\epsilon})$ edges, such that shortcuts constructed by \cite{UllmanY91} have certification complexity $\Omega(m n^\epsilon)$ with high probability.
\end{lemma}

\begin{lemma}
\label{lem:noncert-uy-hs}
For $\epsilon \in [0,1/25)$, there exists a family of unweighted undirected graphs with $n$ vertices and $m=\Theta(n^{1+\epsilon})$ edges, such that hopsets constructed by \cite{UllmanY91} have certification complexity $\Omega(m n^\epsilon)$ with high probability.
\end{lemma}

\begin{lemma}
\label{lem:noncert-kp}
For $\epsilon \in [0,2/75)$, there exists a family of DAGs with $n$ vertices and $m=\Theta(n^{1+\epsilon})$ edges, such that shortcuts constructed by \cite{KoganP22} have certification complexity $\Omega(m n^\epsilon)$ with high probability.
\end{lemma}

\cref{lem:rp} summarizes all necessary properties of the graph construction we use to prove all aforementioned results.
The proof of \cref{lem:rp} is deferred to \cref{sec:proof-rp}.

\begin{restatable}[]{lemma}{lemrp}
\label{lem:rp}
For $\epsilon \in [0,1/25)$, there exists a family of layered DAGs with $n$ vertices and $m=\Theta(n^{1+\epsilon})$ edges such that:
\begin{enumerate}
    \item \label{item:rp-D} It has $D=n^{1/7-\Theta(\epsilon)}$ layers.
    \item \label{item:rp-S} The first and the last layers each have $n^{3/7+\Theta(\epsilon)} \le n^{1/2}$ vertices.
    \item \label{item:rp-P} It has a set $\cP$ of $\Theta(n^{1+2\epsilon}/D)$ critical paths such that:
        \begin{enumerate}
            \item \label{item:rp-P-length} Each path is from the first layer to the last layer.
            \item \label{item:rp-P-unique} Each path is the unique path between its endpoints.
            \item \label{item:rp-P-overlap} The intersection between any two paths is at most one edge.
        \end{enumerate}
\end{enumerate}
\end{restatable}

We now prove all claimed lower bounds using \cref{lem:rp}.

\begin{figure}[tp]
    \centering
    \input{fig/uy}
    \caption{An illustration of the graph construction for \cref{lem:noncert-uy-sc} (red = sampled).}
    \label{fig:uy}
\end{figure}
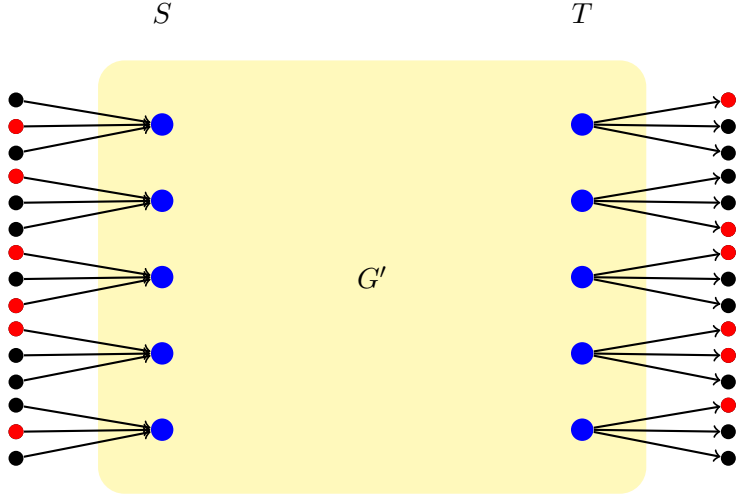

\begin{proof}[Proof of \cref{lem:noncert-uy-sc}]
Let $G'$ be the graph promised by \cref{lem:rp} with $n'=n/3$ vertices and $D'$ layers.
Let $S,T$ be the sets of vertices in the first and the last layers respectively.
So $|S|,|T| < \cO(n^{1/2})$ by \cref{item:rp-S}.
For each $s \in S$, connect $n'/|S| > \Omega(n^{1/2})$ additional auxiliary vertices to $s$.
Also, for each $t \in T$, connect $t$ to $n'/|T| > \Omega(n^{1/2})$ additional auxiliary vertices.
Let $G$ be this new graph, as illustrated in \cref{fig:uy}.
The total number of vertices in $G$ is thus $n=3n'$ and the total number of edges in $G$ is $m=\Theta((n')^{1+\epsilon}+n) = \Theta(n^{1+\epsilon})$.

Recall that \cite{UllmanY91} samples each vertex with probability $\tilde{\Theta}(n^{-1/2})$, and that the constructed shortcut $H$ adds a new edge from every sampled vertex to any other sampled and reachable vertex.
When executed on $G$, by Chernoff bounds and a union bound, \cite{UllmanY91} samples at least one associated auxiliary vertex for every $s \in S$ and every $t \in T$, at the same time with high probability.
It suffices to show \cite{UllmanY91} always constructs shortcuts $H$ with certification complexity $\Omega(m n^\epsilon)$ conditioned on such event happening.

To this end, consider any critical path $P'$ in $G'$.
Let $s \in S, t \in T$ be the two endpoints of $P'$, and let $u,v$ be the sampled auxiliary vertices associated to $s,t$ respectively.
By \cref{item:rp-D,item:rp-P-length,item:rp-P-unique} of \cref{lem:rp}, there is an extended path $P$ from $u$ to $v$ with $\Theta(D')$ edges, which is also the unique $u$-$v$ path.
As a result, in order to add the edge $(u,v)$ to $H$, any certified shortcut $H' \supseteq H$ must contain at least $\Omega(D')$ new edges to certify $(u,v)$.
Moreover, all these added edges can only be used for certifying $(u,v)$ and thus are disjoint from any other added edge for another critical path.
This is because \cref{item:rp-P-overlap} implies that $P$ shares at most three edges with any other extended path (i.e., $(u,s)$, $(t,v)$, and at most one of the edges on $P'$).
Overall, by \cref{item:rp-P}, any certified shortcut $H' \supseteq H$ requires at least $\Theta((n')^{1+2\epsilon}/D') \cdot \Omega(D') = \Omega(n^{1+2\epsilon}) = \Omega(m n^\epsilon)$ new edges, as claimed.
\end{proof}

\begin{proof}[Proof of \cref{lem:noncert-uy-hs}]
This is a direct corollary of the proof of \cref{lem:noncert-uy-sc}.
Concretely, the graph construction promised by \cref{lem:rp} is always a layered DAG.
So, by eliminating the direction of each edge, we immediately get a layered and unweighted undirected graph with exactly the same properties except that the guarantee of \cref{item:rp-P-unique} becomes the uniqueness of shortest paths, instead of arbitrary paths.
The lemma then follows from a similar argument to the one for \cref{lem:noncert-uy-sc}.
\end{proof}

\begin{figure}[tp]
    \centering
    \input{fig/brr}
    \caption{An illustration of the graph construction for \cref{rmk:uy,lem:noncert-brr-sc} (red = sampled).}
    \label{fig:brr}
\end{figure}
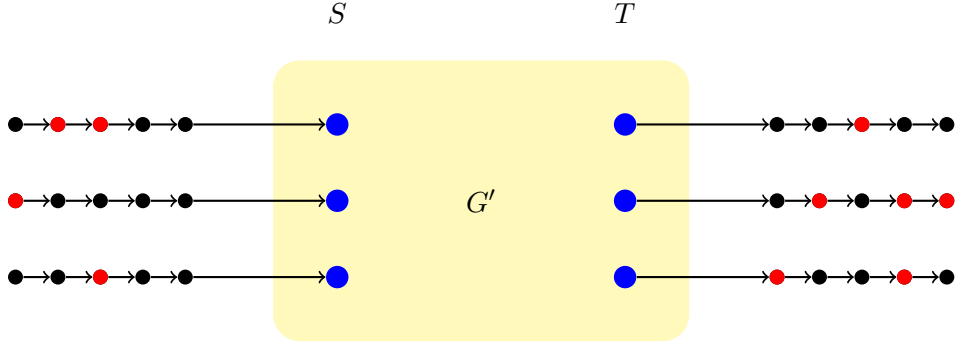

\begin{remark}[Variants of \cite{UllmanY91}]
\label{rmk:uy}
One may have noticed that graphs promised by \cref{lem:rp} only has $o(n^{1/2})$ layers.
It is tempting to modify the approach of \cite{UllmanY91} so that it connects a sampled vertex to another sampled and reachable vertex only if the distance between them is $\Omega(n^{1/2})$ in the original graph (although it is even unclear how to implement such checks efficiently).
Nevertheless, \cref{lem:noncert-uy-sc,lem:noncert-uy-hs} continue to hold against such modification, via a simple tweak in their proofs.
Indeed, for each vertex $s \in S$ in the first layer, we can connect $s$ and all its $\Omega(n^{1/2})$ auxiliary vertices into a path, instead of a star as in \cref{fig:uy}; see \cref{fig:brr}.
The same is done to $T$.
Then, with high probability for every $s \in S$ and every $t \in T$ simultaneously, it samples at least one associated auxiliary vertex that is $\Omega(n^{1/2})$ far away.
Thus, for any critical path $P'$ in $G'$ with endpoints $s,t$, the distance between the sampled auxiliary vertices associated to $s,t$ is always guaranteed to be $\Omega(n^{1/2})$.
Overall, the aforementioned modification does not break the lower bound.
\end{remark}

\begin{figure}[tp]
    \centering
    \input{fig/kp}
    \caption{An illustration of the graph construction for \cref{lem:noncert-kp} (red = sampled).}
    \label{fig:kp}
\end{figure}
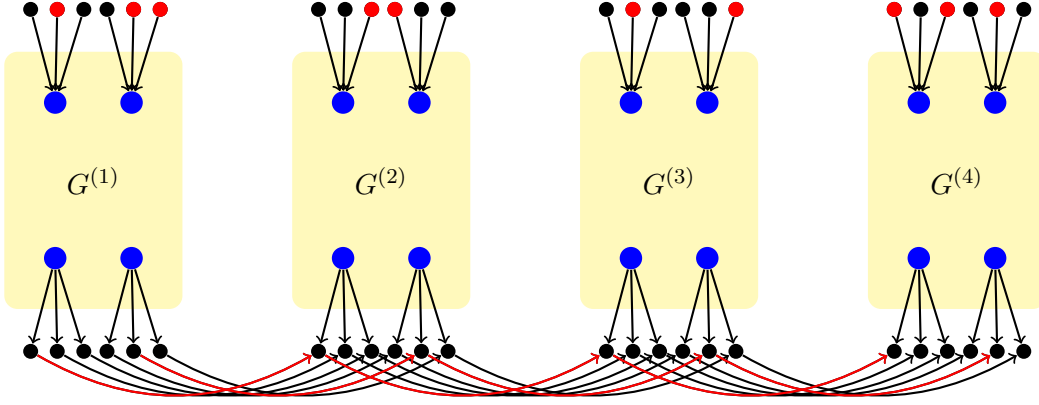

\begin{proof}[Proof of \cref{lem:noncert-kp}]
Let $G'$ be the graph promised by \cref{lem:rp} with $n'=\Theta(n^{2/3})$ vertices and $D'$ layers.
Let $S,T$ be the sets of vertices in the first and the last layers respectively.
So $|S|,|T| < \cO(n^{1/3})$ by \cref{item:rp-S}.
Let $G$ contain $n''=n/(3n')=\Theta(n^{1/3})$ disjoint copies of $G'$.
We write $v^{(i)}$ to denote the vertex corresponding to $v$ in the $i$-th copy $G^{(i)}$.
For $s \in S$ and $i \in [n'']$, connect $n/(3n''|S|) > \Omega(n^{1/3})$ additional auxiliary vertices to $s^{(i)}$.
Also, for $t \in T$ and $i \in [n'']$, connect $t^{(i)}$ to $n/(3n''|T|) > \Omega(n^{1/3})$ additional auxiliary vertices $v^{(i)}_j$ for $j \in [n/(3n''|T|)]$.
In the meantime, connect $v^{(1)}_j,\dots,v^{(n'')}_j$ into a path in this order.
Overall, $n/(3n''|T|) > \Omega(n^{1/3})$ auxiliary paths are associated to every $t$, and all $n/(3n'') = \Theta(n^{2/3})$ auxiliary paths are vertex-disjoint and have $\Theta(n^{1/3})$ edges each.
In total, $G$ has $n=3n'n''$ vertices and $m=\Theta((n')^{1+\epsilon} \cdot n'' + n) = \Theta(n^{1+2\epsilon/3})$ edges.
\cref{fig:kp} demonstrates the constructed graph.

Recall that \cite{KoganP22} samples each vertex and each path with probability $\tilde{\Theta}(n^{-1/3})$, and that the constructed shortcut $H$ adds a new edge from every sampled vertex to its first reachable vertex on every sampled path.
When executed on $G$, by Chernoff bounds and a union bound, \cite{KoganP22} samples at least one associated auxiliary vertex for every $s^{(i)}$ and samples at least one associated auxiliary path for every $t$, at the same time with high probability.
It suffices to show \cite{KoganP22} always constructs shortcuts $H$ with certification complexity $\Omega(m n^\epsilon)$ conditioned on such event happening.

To this end, consider any critical path $P'$ in $G'$.
Let $s \in S, t \in T$ be the two endpoints of $P'$.
Fix $i \in [n'']$.
Let $u$ be the sampled auxiliary vertex associated to $s^{(i)}$.
Also let $(v^{(1)},\dots,v^{(n'')})$ be the sampled auxiliary path associated to $t$.
It is not hard to verify that $v^{(i)}$ is the first vertex on this path reachable from $u$ (via $s^{(i)},t^{(i)}$).
Since all copies of $G'$ are disjoint, by \cref{item:rp-D,item:rp-P-length,item:rp-P-unique} of \cref{lem:rp}, there is an extended path $P$ from $u$ to $v^{(i)}$ with $\Theta(D')$ edges, which is also the unique $u$-$v^{(i)}$ path.
As a result, in order to add the edge $(u,v^{(i)})$ to $H$, any certified shortcut $H' \supseteq H$ must contain at least $\Omega(D')$ new edges to certify $(u,v^{(i)})$.
Moreover, all these added edges can only be used for certifying $(u,v^{(i)})$ and thus are disjoint from any other added edge for another critical path.
This is because \cref{item:rp-P-overlap} implies that $P$ shares at most three edges with any other extended path (i.e., $(u,s^{(i)})$, $(t^{(i)},v^{(i)})$, and at most one of the edges on $P'$).
Overall, by \cref{item:rp-P}, any certified shortcut $H' \supseteq H$ requires at least $\Theta((n')^{1+2\epsilon}/D') \cdot \Omega(D') \cdot n'' = \Omega(n^{1+4\epsilon/3}) = \Omega(m n^{2\epsilon/3})$ new edges.
Substituting $\epsilon$ with $3\epsilon/2$ concludes the proof.
\end{proof}

\begin{remark}[\bf Variants of \cite{KoganP22}]
\label{rmk:kp}
Similar to \cref{rmk:uy}, the proof of \cref{lem:noncert-kp} can be made to have $\Omega(n^{1/3})$ layers in the constructed graphs.
Meanwhile, it is known that the existential proof of \cite{KoganP22} can be adapted to sample only paths.
Specifically, after sampling each path with probability $\tilde{\Theta}(n^{-1/3})$, the constructed shortcut adds a new edge from every vertex on every sampled path to its first reachable vertex on every other sampled path.
Our proof is robust to such adaptation and \cref{lem:noncert-kp} continues to hold for this variant.
In fact, simply connect the auxiliary vertices associated to $S$ into auxiliary paths, in the same way as what we do for $T$ in the proof of \cref{lem:noncert-kp}.
Then, the lower bound follows from a similar argument.
\end{remark}

\subsubsection{Greedy Approaches}
\label{sec:existing-noncert-greedy}

In a nutshell, the greedy approach of \cite{BermanRR10} iteratively adds the edge that reduces the potential the most, where the potential is defined as the sum of all pairwise distances.
As such, \cite{BermanRR10} completely avoids the use of random sampling.
Meanwhile, it slightly improves upon the guarantee of \cite{UllmanY91} by eliminating logarithmic factors.
In \cref{lem:noncert-brr-sc,lem:noncert-brr-hs}, however, we prove the same lower bounds on its certification complexity as for \cite{UllmanY91} in \cref{lem:noncert-uy-sc,lem:noncert-uy-hs}.

Recently, \cite{bals2025greedy} (Section 3) obtains a greedy approach based on \cite{KoganP22} that constructs $\cO(n^{1/3})$-shortcuts of size $\cO(n)$ as well.
The same lower bound as for \cite{KoganP22} also holds for \cite{bals2025greedy}.

\begin{lemma}
\label{lem:noncert-brr-sc}
For $\epsilon \in [0,1/25)$, there exists a family of DAGs with $n$ vertices and $m=\Theta(n^{1+\epsilon})$ edges, such that shortcuts constructed by \cite{BermanRR10} have certification complexity $\Omega(m n^\epsilon)$ with high probability.
\end{lemma}

\begin{lemma}
\label{lem:noncert-brr-hs}
For $\epsilon \in [0,1/25)$, there exists a family of unweighted undirected graphs with $n$ vertices and $m=\Theta(n^{1+\epsilon})$ edges, such that hopsets constructed by \cite{BermanRR10} have certification complexity $\Omega(m n^\epsilon)$ with high probability.
\end{lemma}

\begin{lemma}
\label{lem:noncert-bbb}
For $\epsilon \in [0,2/75)$, there exists a family of DAGs with $n$ vertices and $m=\Theta(n^{1+\epsilon})$ edges, such that shortcuts constructed by the greedy approach of \cite{bals2025greedy} have certification complexity $\Omega(m n^\epsilon)$ with high probability.
\end{lemma}

We prove \cref{lem:noncert-brr-sc} in the following.
Proofs of \cref{lem:noncert-brr-hs,lem:noncert-bbb} are omitted as they follow from the proof of \cref{lem:noncert-brr-sc} in a similar way that \cref{lem:noncert-uy-hs,lem:noncert-kp} follow from \cref{lem:noncert-uy-sc}.

\begin{proof}[Proof of \cref{lem:noncert-brr-sc}]
We use the same graph $G$ as in \cref{rmk:uy}, as illustrated in \cref{fig:brr}.
That is, take the graph $G'$ promised by \cref{lem:rp} and connect a disjoint path of length $\Omega(n^{1/2})$ to each vertex in the first and the last layers.
For any critical path $P'$ in $G'$ with endpoints $s,t$, note that $P'$ has less than $n^{1/2}$ edges and that it is now extended by $\Omega(n^{1/2})$ auxiliary edges on both ends.

Recall that \cite{BermanRR10} adds the edge that reduces the potential the most, where the potential is defined as the sum of all pairwise distances.
So the greedy approach must connect two vertices $u,v$ which divide the extended path into three subpaths of length $\Omega(n^{1/2})$ each.
In particular, it means that $u,v$ must be auxiliary vertices associated to $s,t$ respectively since $P'$ has less than $n^{1/2}$ edges.
Furthermore, the greedy approach has to repeat this for every $P'$.
As a result, the lemma follows from a similar argument to the one for \cref{lem:noncert-uy-sc}.
\end{proof}

\begin{remark}[\bf Variants of \cite{BermanRR10}]
\label{rmk:brr}
For simplicity, the proof above is against the approach of \cite{BermanRR10} as described by \cite{bals2025greedy} (Section 4).
Specifically, each iteration adds the edge that reduces the potential the most.
The original version can be viewed as adding any edge that reduces the potential by $\Omega(D^3)$ whenever the diameter of the graph is $D$.
\cref{lem:noncert-brr-sc} continues to hold under this description.
In fact, the only concern would be that the auxiliary paths might be shortcut to diameter $o(n^{1/2})$ before shortcutting the critical paths.
(Otherwise, the same proof works.)
This is actually impossible.
To see this, suppose the auxiliary path associated to $s \in S$ is shortcut to diameter $o(n^{1/2})$.
Let $s'$ be the other end of the auxiliary path, and $s''$ be the closest auxiliary vertex to $s'$ that has ever been shortcut so far.
Then, the distance between $s'$ and $s''$ is $o(n^{1/2})$.
However, this implies that the potential drops by at most $o(n^{3/2})$ in the iteration shortcutting $s''$.
This is a contradiction as the diameter of the graph is always $\Omega(n^{1/2})$.
\end{remark}

\subsubsection{Non-Constructive Chain Covers}
\label{sec:existing-noncert-cc}

We prove the existence of non-constructive chain covers in this section.

\begin{lemma}
\label{lem:noncert-cc}
For $\epsilon \in [0,1/50)$, there exists a family of DAGs $G$ with $n$ vertices, $m=\Theta(n^{1+\epsilon})$ edges, and $\Theta(n^{1/2})$-chain covers $\cC$ such that any shortcut containing edges of all chains in $\cC$ has certification complexity $\Omega(m n^\epsilon)$.
\end{lemma}

\begin{figure}[tp]
    \centering
    \input{fig/cc}
    \caption{An illustration of the graph construction for \cref{lem:noncert-cc} (red: a possible chain).}
    \label{fig:cc}
\end{figure}
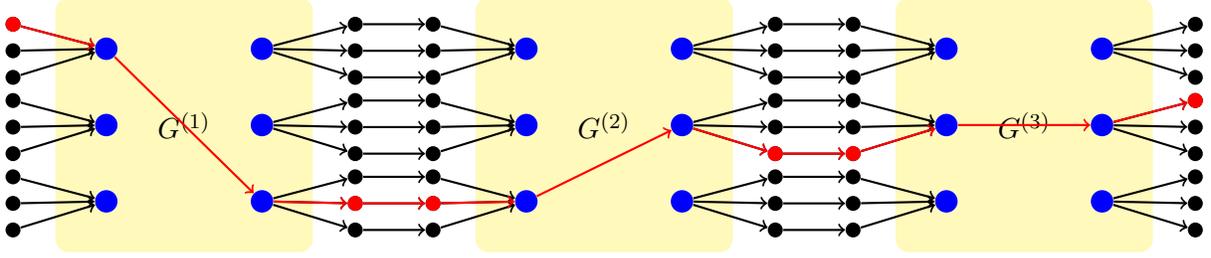

\begin{proof}
Let $G'$ be the graph promised by \cref{lem:noncert-uy-sc} with $n'=n^{1/2}$ vertices and $m'=\Theta((n')^{1+\epsilon})=\Theta(n^{1/2+\epsilon/2})$ edges.
Let $S,T$ be the sets of vertices in the first and the last non-auxiliary layers respectively.
By the proofs of \cref{lem:noncert-uy-sc,lem:rp}, we have $|S|=|T|=\cO((n')^{1/2})=\cO(n^{1/4})$.
Meanwhile, $\Delta=\cO(|S|^{1-o(1)})=\cO(n^{1/4-o(1)})$ critical paths start from each $s \in S$ and $\Delta$ critical paths end at each $t \in T$.
Let $G$ contain $n/n'=n^{1/2}$ disjoint copies of $G'$.
We use superscripts $i \in [n/n']$ to denote corresponding objects in different copies.
For $i > 1$, the auxiliary vertices associated to $T^{(i-1)}$ are connected to the auxiliary vertices associated to $S^{(i)}$ by a canonical perfect bipartite matching, as depicted in \cref{fig:cc}.
So $G$ has $n = n' \cdot n/n'$ vertices and $m = \Theta(m' \cdot n/n' + n) = \Theta(n^{1+\epsilon/2})$ edges in total.

In the following, we construct a subset of chains $\cC' \subseteq \cC$ such that even certifying edges of all chains in $\cC'$ alone already requires $\Omega(n^{1+\epsilon})=\Omega(m n^{\epsilon/2})$ new edges.
The lemma then follows by substituting $\epsilon$ with $2\epsilon$.
To this end, we gradually augment the chains by adding vertices from $G^{(i)}$ in the increasing order of $i$.

Initially, $i=1$.
For each vertex $s \in S$, there are $\Delta$ chains starting from distinct auxiliary vertices associated to $s^{(1)}$.
This is well-defined because there are at least $\Omega((n')^{1/2})=\Omega(n^{1/4}) \ge \Delta$ different associated auxiliary vertices by the proof of \cref{lem:noncert-uy-sc}.
Each of these $\Delta$ chains will be augmented according to a unique one of the $\Delta$ critical paths starting from $s$.
In particular, let $C$ be the chain to be augmented and $P$ the critical path used.
Due to the properties of critical paths, $P$ must end at a unique vertex $t \in T$.
Then add to $C$ an auxiliary vertex associated to $t^{(1)}$.
Note that over all starting vertices $s \in S$, a total of $\Delta$ critical paths end at $t$.
So it is possible to ensure that distinct auxiliary vertices associated to $t^{(1)}$ are added to the chains each time.

Now observe that the above process constructs $|S| \cdot \Delta = \cO(n^{1/2-o(1)})$ vertex-disjoint chains.
Each of them exactly contains one auxiliary vertex associated to some $s^{(1)} \in S^{(1)}$ and one auxiliary vertex associated to some $t^{(1)} \in T^{(1)}$.
More importantly, each of them corresponds to a unique critical path in $G^{(1)}$.
Altogether, in order to certify the edges of all these chains, all critical paths must be certified.
Using an argument similar to the one for \cref{lem:noncert-uy-sc}, this means at least $\Omega(m'(n')^\epsilon)=\Omega(m n^{\epsilon/2} / (n/n'))$ new edges have to be added.

Finally, we can make all chains follow the unique edge going from $G^{(1)}$ to $G^{(2)}$.
So the above argument can be repeated for all copies $i \in [n/n']$.
In total, certifying all chain edges needs at least $\Omega(m n^{\epsilon/2})$ added edges, as claimed.
\end{proof}

\begin{remark}
In fact, the sampling approach of \cite{KoganP22} and its greedy version \cite{bals2025greedy} all need to shortcut certain chain covers as well.
It is not hard to see that the proof of \cref{lem:noncert-cc} can be modified to show the existence of non-constructive $\Theta(n^c)$-chain covers for basically any $c \in (0,1)$.
Thus, all these approaches depend on the quality of the chain covers.
This is in addition to the results we prove in \cref{sec:existing-noncert-sample,sec:existing-noncert-greedy} which essentially imply that the constructed shortcuts have high certification complexity even if the chain edges are given for free.
\end{remark}

%% file: fig/uy.tex
\begin{tikzpicture}
    \begin{scope}[local bounding box=bb]
        \node[fill=blue,circle,inner sep=3pt] (s1) {};
        \foreach \i in {2,...,5}
        {
            \pgfmathtruncatemacro{\ii}{\i-1}
            \node[fill=blue,circle,inner sep=3pt,below=20pt of s\ii] (s\i) {};
        }

        \node[fill=blue,circle,inner sep=3pt,right=150pt of s1] (t1) {};
        \foreach \i in {2,...,5}
        {
            \pgfmathtruncatemacro{\ii}{\i-1}
            \node[fill=blue,circle,inner sep=3pt,below=20pt of t\ii] (t\i) {};
        }
    \end{scope}

    \foreach \i in {1,...,5}
    {
        \node[fill,circle,inner sep=2pt,above left=4pt and 50pt of s\i] (s\i1) {};
        \draw[->,thick] (s\i1) -- (s\i);
        \foreach \j in {2,...,3}
        {
            \pgfmathtruncatemacro{\jj}{\j-1}
            \node[fill,circle,inner sep=2pt,below=4pt of s\i\jj] (s\i\j) {};
        \draw[->,thick] (s\i\j) -- (s\i);
        }
    }
    
    \foreach \i in {1,...,5}
    {
        \node[fill,circle,inner sep=2pt,above right=4pt and 50pt of t\i] (t\i1) {};
        \draw[->,thick] (t\i) -- (t\i1);
        \foreach \j in {2,...,3}
        {
            \pgfmathtruncatemacro{\jj}{\j-1}
            \node[fill,circle,inner sep=2pt,below=4pt of t\i\jj] (t\i\j) {};
        \draw[->,thick] (t\i) -- (t\i\j);
        }
    }

    \begin{scope}[on background layer]
        \fill[yellow!33,rounded corners=10pt] ($(bb.south west) - (20pt,20pt)$) rectangle ($(bb.north east) + (20pt,20pt)$);
        \node at (bb.center) {$G'$};
        \node[above=30pt of s1] {$S$};
        \node[above=30pt of t1] {$T$};
    \end{scope}

    \node[fill=red,circle,inner sep=2pt] at (s12) {};
    \node[fill=red,circle,inner sep=2pt] at (s21) {};
    \node[fill=red,circle,inner sep=2pt] at (s31) {};
    \node[fill=red,circle,inner sep=2pt] at (s33) {};
    \node[fill=red,circle,inner sep=2pt] at (s41) {};
    \node[fill=red,circle,inner sep=2pt] at (s52) {};
    \node[fill=red,circle,inner sep=2pt] at (t11) {};
    \node[fill=red,circle,inner sep=2pt] at (t23) {};
    \node[fill=red,circle,inner sep=2pt] at (t31) {};
    \node[fill=red,circle,inner sep=2pt] at (t41) {};
    \node[fill=red,circle,inner sep=2pt] at (t42) {};
    \node[fill=red,circle,inner sep=2pt] at (t51) {};
\end{tikzpicture}

%% file: fig/brr.tex
\begin{tikzpicture}
    \begin{scope}[local bounding box=bb]
        \node[fill=blue,circle,inner sep=3pt] (s1) {};
        \foreach \i in {2,...,3}
        {
            \pgfmathtruncatemacro{\ii}{\i-1}
            \node[fill=blue,circle,inner sep=3pt,below=20pt of s\ii] (s\i) {};
        }

        \node[fill=blue,circle,inner sep=3pt,right=100pt of s1] (t1) {};
        \foreach \i in {2,...,3}
        {
            \pgfmathtruncatemacro{\ii}{\i-1}
            \node[fill=blue,circle,inner sep=3pt,below=20pt of t\ii] (t\i) {};
        }
    \end{scope}

    \foreach \i in {1,...,3}
    {
        \node[fill,circle,inner sep=2pt,left=50pt of s\i] (s\i1) {};
        \draw[->,thick] (s\i1) -- (s\i);
        \foreach \j in {2,...,5}
        {
            \pgfmathtruncatemacro{\jj}{\j-1}
            \node[fill,circle,inner sep=2pt,left=10pt of s\i\jj] (s\i\j) {};
        \draw[->,thick] (s\i\j) -- (s\i\jj);
        }
    }
    
    \foreach \i in {1,...,3}
    {
        \node[fill,circle,inner sep=2pt,right=50pt of t\i] (t\i1) {};
        \draw[->,thick] (t\i) -- (t\i1);
        \foreach \j in {2,...,5}
        {
            \pgfmathtruncatemacro{\jj}{\j-1}
            \node[fill,circle,inner sep=2pt,right=10pt of t\i\jj] (t\i\j) {};
        \draw[->,thick] (t\i\jj) -- (t\i\j);
        }
    }

    \begin{scope}[on background layer]
        \fill[yellow!33,rounded corners=10pt] ($(bb.south west) - (20pt,20pt)$) rectangle ($(bb.north east) + (20pt,20pt)$);
        \node at (bb.center) {$G'$};
        \node[above=30pt of s1] {$S$};
        \node[above=30pt of t1] {$T$};
    \end{scope}

    \node[fill=red,circle,inner sep=2pt] at (s13) {};
    \node[fill=red,circle,inner sep=2pt] at (s14) {};
    \node[fill=red,circle,inner sep=2pt] at (s25) {};
    \node[fill=red,circle,inner sep=2pt] at (s33) {};
    \node[fill=red,circle,inner sep=2pt] at (t13) {};
    \node[fill=red,circle,inner sep=2pt] at (t22) {};
    \node[fill=red,circle,inner sep=2pt] at (t24) {};
    \node[fill=red,circle,inner sep=2pt] at (t25) {};
    \node[fill=red,circle,inner sep=2pt] at (t31) {};
    \node[fill=red,circle,inner sep=2pt] at (t34) {};
\end{tikzpicture}

%% file: fig/kp.tex
\begin{tikzpicture}
    \begin{scope}[local bounding box=bb1]
        \node[fill=blue,circle,inner sep=3pt] (s11) {};
        \foreach \i in {2,...,2}
        {
            \pgfmathtruncatemacro{\ii}{\i-1}
            \node[fill=blue,circle,inner sep=3pt,right=20pt of s1\ii] (s1\i) {};
        }

        \node[fill=blue,circle,inner sep=3pt,below=50pt of s11] (t11) {};
        \foreach \i in {2,...,2}
        {
            \pgfmathtruncatemacro{\ii}{\i-1}
            \node[fill=blue,circle,inner sep=3pt,right=20pt of t1\ii] (t1\i) {};
        }
    \end{scope}

    \foreach \i in {1,...,2}
    {
        \node[fill,circle,inner sep=2pt,above left=30pt and 4pt of s1\i] (ss1\i1) {};
        \draw[->,thick] (ss1\i1) -- (s1\i);
        \foreach \j in {2,...,3}
        {
            \pgfmathtruncatemacro{\jj}{\j-1}
            \node[fill,circle,inner sep=2pt,right=4pt of ss1\i\jj] (ss1\i\j) {};
        \draw[->,thick] (ss1\i\j) -- (s1\i);
        }
    }
    
    \foreach \i in {1,...,2}
    {
        \node[fill,circle,inner sep=2pt,below left=30pt and 4pt of t1\i] (tt1\i1) {};
        \draw[->,thick] (t1\i) -- (tt1\i1);
        \foreach \j in {2,...,3}
        {
            \pgfmathtruncatemacro{\jj}{\j-1}
            \node[fill,circle,inner sep=2pt,right=4pt of tt1\i\jj] (tt1\i\j) {};
        \draw[->,thick] (t1\i) -- (tt1\i\j);
        }
    }

    \begin{scope}[on background layer]
        \fill[yellow!33,rounded corners=5pt] ($(bb1.south west) - (15pt,15pt)$) rectangle ($(bb1.north east) + (15pt,15pt)$);
        \node at (bb1.center) {$G^{(1)}$};
    \end{scope}

    \foreach \k in {2,...,4}
    {
        \pgfmathtruncatemacro{\kk}{\k-1}
        \begin{scope}[local bounding box=bb\k]
            \node[fill=blue,circle,inner sep=3pt,right=100pt of s\kk1] (s\k1) {};
            \foreach \i in {2,...,2}
            {
                \pgfmathtruncatemacro{\ii}{\i-1}
                \node[fill=blue,circle,inner sep=3pt,right=20pt of s\k\ii] (s\k\i) {};
            }
    
            \node[fill=blue,circle,inner sep=3pt,below=50pt of s\k1] (t\k1) {};
            \foreach \i in {2,...,2}
            {
                \pgfmathtruncatemacro{\ii}{\i-1}
                \node[fill=blue,circle,inner sep=3pt,right=20pt of t\k\ii] (t\k\i) {};
            }
        \end{scope}
    
        \foreach \i in {1,...,2}
        {
            \node[fill,circle,inner sep=2pt,above left=30pt and 4pt of s\k\i] (ss\k\i1) {};
            \draw[->,thick] (ss\k\i1) -- (s\k\i);
            \foreach \j in {2,...,3}
            {
                \pgfmathtruncatemacro{\jj}{\j-1}
                \node[fill,circle,inner sep=2pt,right=4pt of ss\k\i\jj] (ss\k\i\j) {};
            \draw[->,thick] (ss\k\i\j) -- (s\k\i);
            }
        }
        
        \foreach \i in {1,...,2}
        {
            \node[fill,circle,inner sep=2pt,below left=30pt and 4pt of t\k\i] (tt\k\i1) {};
            \draw[->,thick] (t\k\i) -- (tt\k\i1);
            \foreach \j in {2,...,3}
            {
                \pgfmathtruncatemacro{\jj}{\j-1}
                \node[fill,circle,inner sep=2pt,right=4pt of tt\k\i\jj] (tt\k\i\j) {};
            \draw[->,thick] (t\k\i) -- (tt\k\i\j);
            }
        }
    
        \begin{scope}[on background layer]
            \fill[yellow!33,rounded corners=5pt] ($(bb\k.south west) - (15pt,15pt)$) rectangle ($(bb\k.north east) + (15pt,15pt)$);
            \node at (bb\k.center) {$G^{(\k)}$};
        \end{scope}
    }

    \foreach \k in {2,...,4}
    {
        \pgfmathtruncatemacro{\kk}{\k-1}
        \foreach \i in {1,...,2}
        {
            \foreach \j in {1,...,3}
            {
                \draw[->,thick,bend right] (tt\kk\i\j) edge (tt\k\i\j);
            }
        }
    }

    \node[fill=red,circle,inner sep=2pt] at (ss112) {};
    \node[fill=red,circle,inner sep=2pt] at (ss122) {};
    \node[fill=red,circle,inner sep=2pt] at (ss123) {};
    \node[fill=red,circle,inner sep=2pt] at (ss213) {};
    \node[fill=red,circle,inner sep=2pt] at (ss221) {};
    \node[fill=red,circle,inner sep=2pt] at (ss312) {};
    \node[fill=red,circle,inner sep=2pt] at (ss323) {};
    \node[fill=red,circle,inner sep=2pt] at (ss411) {};
    \node[fill=red,circle,inner sep=2pt] at (ss413) {};
    \node[fill=red,circle,inner sep=2pt] at (ss422) {};
    
    \draw[->,thick,bend right,red] (tt111) edge (tt211);
    \draw[->,thick,bend right,red] (tt211) edge (tt311);
    \draw[->,thick,bend right,red] (tt311) edge (tt411);
    \draw[->,thick,bend right,red] (tt122) edge (tt222);
    \draw[->,thick,bend right,red] (tt222) edge (tt322);
    \draw[->,thick,bend right,red] (tt322) edge (tt422);
\end{tikzpicture}

%% file: fig/cc.tex
\begin{tikzpicture}
    \begin{scope}[local bounding box=bb1]
        \node[fill=blue,circle,inner sep=3pt] (s11) {};
        \foreach \i in {2,...,3}
        {
            \pgfmathtruncatemacro{\ii}{\i-1}
            \node[fill=blue,circle,inner sep=3pt,below=20pt of s1\ii] (s1\i) {};
        }

        \node[fill=blue,circle,inner sep=3pt,right=50pt of s11] (t11) {};
        \foreach \i in {2,...,3}
        {
            \pgfmathtruncatemacro{\ii}{\i-1}
            \node[fill=blue,circle,inner sep=3pt,below=20pt of t1\ii] (t1\i) {};
        }
    \end{scope}

    \foreach \i in {1,...,3}
    {
        \node[fill,circle,inner sep=2pt,above left=4pt and 30pt of s1\i] (ss1\i1) {};
        \draw[->,thick] (ss1\i1) -- (s1\i);
        \foreach \j in {2,...,3}
        {
            \pgfmathtruncatemacro{\jj}{\j-1}
            \node[fill,circle,inner sep=2pt,below=4pt of ss1\i\jj] (ss1\i\j) {};
        \draw[->,thick] (ss1\i\j) -- (s1\i);
        }
    }
    
    \foreach \i in {1,...,3}
    {
        \node[fill,circle,inner sep=2pt,above right=4pt and 30pt of t1\i] (tt1\i1) {};
        \draw[->,thick] (t1\i) -- (tt1\i1);
        \foreach \j in {2,...,3}
        {
            \pgfmathtruncatemacro{\jj}{\j-1}
            \node[fill,circle,inner sep=2pt,below=4pt of tt1\i\jj] (tt1\i\j) {};
        \draw[->,thick] (t1\i) -- (tt1\i\j);
        }
    }

    \begin{scope}[on background layer]
        \fill[yellow!33,rounded corners=5pt] ($(bb1.south west) - (15pt,15pt)$) rectangle ($(bb1.north east) + (15pt,15pt)$);
        \node at (bb1.center) {$G^{(1)}$};
    \end{scope}

    \foreach \k in {2,...,3}
    {
        \pgfmathtruncatemacro{\kk}{\k-1}
        \begin{scope}[local bounding box=bb\k]
            \node[fill=blue,circle,inner sep=3pt,right=150pt of s\kk1] (s\k1) {};
            \foreach \i in {2,...,3}
            {
                \pgfmathtruncatemacro{\ii}{\i-1}
                \node[fill=blue,circle,inner sep=3pt,below=20pt of s\k\ii] (s\k\i) {};
            }
    
            \node[fill=blue,circle,inner sep=3pt,right=50pt of s\k1] (t\k1) {};
            \foreach \i in {2,...,3}
            {
                \pgfmathtruncatemacro{\ii}{\i-1}
                \node[fill=blue,circle,inner sep=3pt,below=20pt of t\k\ii] (t\k\i) {};
            }
        \end{scope}
        
        \foreach \i in {1,...,3}
        {
            \node[fill,circle,inner sep=2pt,above left=4pt and 30pt of s\k\i] (ss\k\i1) {};
            \draw[->,thick] (ss\k\i1) -- (s\k\i);
            \foreach \j in {2,...,3}
            {
                \pgfmathtruncatemacro{\jj}{\j-1}
                \node[fill,circle,inner sep=2pt,below=4pt of ss\k\i\jj] (ss\k\i\j) {};
            \draw[->,thick] (ss\k\i\j) -- (s\k\i);
            }
        }
        
        \foreach \i in {1,...,3}
        {
            \node[fill,circle,inner sep=2pt,above right=4pt and 30pt of t\k\i] (tt\k\i1) {};
            \draw[->,thick] (t\k\i) -- (tt\k\i1);
            \foreach \j in {2,...,3}
            {
                \pgfmathtruncatemacro{\jj}{\j-1}
                \node[fill,circle,inner sep=2pt,below=4pt of tt\k\i\jj] (tt\k\i\j) {};
            \draw[->,thick] (t\k\i) -- (tt\k\i\j);
            }
        }
        
        \begin{scope}[on background layer]
            \fill[yellow!33,rounded corners=5pt] ($(bb\k.south west) - (15pt,15pt)$) rectangle ($(bb\k.north east) + (15pt,15pt)$);
            \node at (bb\k.center) {$G^{(\k)}$};
        \end{scope}
    }
    
    \foreach \k in {2,...,3}
    {
        \pgfmathtruncatemacro{\kk}{\k-1}
        \foreach \i in {1,...,3}
        {
            \foreach \j in {1,...,3}
            {
                \draw[->,thick] (tt\kk\i\j) -- (ss\k\i\j);
            }
        }
    }

    \node[fill=red,circle,inner sep=2pt] at (ss111) {};
    \node[fill=red,circle,inner sep=2pt] at (tt132) {};
    \node[fill=red,circle,inner sep=2pt] at (ss232) {};
    \node[fill=red,circle,inner sep=2pt] at (tt223) {};
    \node[fill=red,circle,inner sep=2pt] at (ss323) {};
    \node[fill=red,circle,inner sep=2pt] at (tt321) {};
    
    \draw[->,thick,red] (ss111) -- (s11);
    \draw[->,thick,red] (s11) -- (t13);
    \draw[->,thick,red] (t13) -- (tt132);
    \draw[->,thick,red] (tt132) -- (ss232);
    \draw[->,thick,red] (ss232) -- (s23);
    \draw[->,thick,red] (s23) -- (t22);
    \draw[->,thick,red] (t22) -- (tt223);
    \draw[->,thick,red] (tt223) -- (ss323);
    \draw[->,thick,red] (ss323) -- (s32);
    \draw[->,thick,red] (s32) -- (t32);
    \draw[->,thick,red] (t32) -- (tt321);
\end{tikzpicture}

%% file: appendix_lowdepth.tex
\section{Extending a Certified Shortcut to have Low Certification Depth}\label{sec:low-depth-certificates}

In this section, we show the following:

\begin{lemma}
    For any $n$-vertex DAG $G = (V, E)$ and certified shortcut $H$ on $G$, there exists a sequence $H_1, H_2, \dots, H_k$ of $k = O(\log n)$ shortcuts of total size $\sum_i |H_i| = O(|H| \log n)$ such that
    \begin{itemize}
        \item $H \subseteq E \cup H_1 \cup \dots \cup H_k$.
        \item For every edge $e = (u, v) \in H_i$, there is a pair of edges $(u, w), (w, v) \in E \cup H_1 \cup \dots H_{i - 1}$ appearing earlier in the sequence that certify $e$.
    \end{itemize}
\end{lemma}

\begin{proof}
    Fix the two certifying edges for every shortcut edge, and let the \textit{certifying path} $P_e$ of each shortcut edge $e \in H$ be the concatenation of the certifying paths $P_{e_1}$ and $P_{e_2}$ of its two certifying edges $e_1, e_2$, where we let the certifying path of an edge $e \in E$ be just that edge itself.
    
    First, assign to each shortcut edge $e \in H$ a binary tree $T_e$ on the edges on its certifying path $P_e$, such that the following hold:
    \begin{enumerate}
        \item Each tree has asymptotically optimal depth $O(\log |P_e|)$.
        \item The total number of \emph{distinct} subtrees over all the trees $T_e$ ($e \in H$) is at most $O(|H| \log n)$.
    \end{enumerate}
    Such an assignment can be achieved by using a standard persistent balanced binary tree structure such as a treap \cite{AragonS89} that supports the join-operation on two trees (this operation combines the two input trees so that in the resulting tree, every node of the first tree appears left of every node of the second tree) with worst-case $O(\log n)$ new nodes created. Then, we simply let $T_e = \mathrm{join}(T_{e_1}, T_{e_2})$ for each shortcut edge $e \in H$.

    Each subtree of any of these binary trees corresponds to a path on $G$. For a subtree rooted at a tree node $x$ that contains the edge $e \in E$, has left child $y$, and has right child $z$, this path $P_x$ is the concatenation of $P_y$, $e$ and $P_z$, where $P_y$ is empty if $x$ has no left child, and $P_z$ is empty if it has no right child. We have $P_e = P_x$ for the root $x$ of $T_e$.

    Now, let $k := 2 \max_{e \in H} \mathrm{depth}(T_e)$ be two times the greatest depth of any of the trees $T_e$. Let $x$ be the root tree node of some unique non-leaf subtree, and let $d > 0$ be the depth of the subtree. We include for each such $x$ a shortcut edge $e_{rep}(x)$ shortcutting its path $P_x$ in $H_{2d}$. If the tree node has both a left child $y$ and a right child $z$, to help certify $e_{rep}(x)$, we additionally include in $H_{2d - 1}$ a dummy edge $e_{dum}(x)$ shortcutting the concatenation of $P_{y}$ and the edge $e \in E$ contained in $x$.
    
    Now,
    \begin{itemize}
        \item Since the graph is a DAG, $|P_e| \leq n$ for each $e \in H$. Thus, $k = O(\log n)$. 
        \item Since for each unique subtree over the trees $T_e$ we included at most two edges into $H_1, \dots, H_k$, and there are $O(|H| \log n)$ unique subtrees, we have $\sum_i |H_i| = O(|H| \log n)$.
        \item Since $P_e = P_x$ for the root $x$ of $T_e$ for each $e \in H$, and the edge $e_{rep}(x)$ shortcuts $P_x$, we have $e_{rep}(x) = e$. Thus, $H \subseteq E \cup H_1 \cup \dots \cup H_k$. 
        \item Let $x$ be the root of some unique subtree. If $x$ has both a left child $y$ and a right child $z$, then
        \begin{itemize}
            \item $e_{rep}(x)$ is certified by $e_{dum}(x)$ and $e_{rep}(z)$. Both of these edges appear earlier than $e_{rep}(x)$ as the subtrees of the children of $x$ have strictly smaller depth than the subtree of $x$.
            \item $e_{dum}(x)$ is certified by $e_{rep}(y)$ and the edge $e \in E$ contained in $x$.
        \end{itemize}
        If $x$ has only a left child $y$, $e_{rep}(x)$ is certified by $e_{rep}(y)$ and the edge $e \in E$ contained in $x$. If $x$ has only a right child $z$, $e_{rep}(x)$ is certified by the edge $e \in E$ contained in $x$ and $e_{rep}(z)$. 
    \end{itemize}  
\end{proof}

%% file: app_proof.tex
\section{Deferred Proofs}
\label{sec:miss-proof}

\subsection{Proof of \cref{lem:hs}}
\label{sec:proof-hs}

{\renewcommand\footnote[1]{}\lemhs*}

\begin{proof}
The graph $G$ is constructed as follows.
For simplicity, we first consider the directed setting by constructing DAGs for any $d>0$.
It is extended to the undirected setting at the end, although only for $d \in \set{1,2}$.

\paragraph{Vertices.}

The vertex set $V$ consists of all $(d+2)$-dimensional grid points on $[d] \times [4Dr]^{d+1}$.
So we have $|V|=\Theta(dD^{d+1} r^{d+1})$, as claimed in \cref{item:hs-N}.

\paragraph{Edges.}

For each $i \in [d]$, $\bu \in [4Dr]^{d+1}$, and $\bv=(y,z) \in \cV(r)$, the edge set $E$ contains an edge from $(i,\bu)$ to $((i \bmod d)+1,\bu+\be_i(\bv))$, if it does not go outside the grid, where $\be_i(\bv)$ is the $(d+1)$-dimensional vector with all zero coordinates except that the $i$-th and the $(i+1)$-th coordinates are $y,z$ respectively.
So we have $|E|=\Theta(|V| \cdot \Delta)$, where $\Delta=|\cV(r)|=\Theta(r^{2/3})$ by \cref{lem:ball-size}, as claimed in \cref{item:hs-N}.

\paragraph{Critical paths.}

\begin{algorithm}[tp]
    \caption{The algorithm to construct $\cP$.}
    \label{alg:path}
    \begin{algorithmic}[1]
        \State Let $\cP \gets \emptyset$.
        \For{$\bv_1,\dots,\bv_d \in \cV(r)$}
            \State Let $E' \gets E$.
            \Repeat
                \State Let $P=((i_1,\bu_1),\dots,(i_{dD+1},\bu_{dD+1}))$ be a path in $E'$ such that for all $j \in [dD]$,
                    \[
                    (i_{j+1},\bu_{j+1}) = ((i_j \bmod d)+1,\bu_j + \be_{i_j}(\bv_{i_j})).
                    \]
                \State Let $\cP \gets \cP \cup \set{P}$ and $E' \gets E' \setminus P$.
            \Until{No $P$ exists.}
        \EndFor
    \end{algorithmic}
\end{algorithm}

The set of critical paths $\cP$ is constructed by \cref{alg:path}.
\cref{item:hs-P-length} follows directly from the construction.

For \cref{item:hs-P-unique}, observe that each $P \in \cP$ uniquely determines its first vertex $(i_1,\bu_1)$ and $d$ direction vectors $\bv_i=(y_i,z_i)$ for $i \in [d]$, and vice versa.
Also, its last vertex is
\[
(i_{dD+1}=i_1, \bu_{dD+1} = \bu_1 + D \cdot \sum_{i \in [d]} \be_i(\bv_i)).
\]
Let $P'$ be any shortest path from $(i_1,\bu_1)$ to $(i_{dD+1},\bu_{dD+1})$ in $G$.
Since $P'$ is a shortest path, we can get $|P'| \le |P| = dD$.
In fact,  $|P'| = dD'$ for some $D' \le D$ due to the construction of the edge set $E$.
So $P'=((i'_1=i_1,\bu'_1=\bu_1),\dots,(i'_{dD'+1}=i_{dD+1},\bu'_{dD'+1}=\bu_{dD+1}))$.
Furthermore, $P'$ induces direction vectors $\bv_{i,j}=(y_{i,j},z_{i,j})$ for $i \in [d]$ and $j \in [D']$ such that
\begin{align}
D \cdot y_1 & = \sum_{j \in [D']} y_{1,j}, \nonumber\\
D \cdot z_i + D \cdot y_{i+1} & = \sum_{j \in [D']} z_{i,j} + \sum_{j \in [D']} y_{i+1,j}, \qquad \forall i \in [d-1] \label{eqn:unique-sum}\\
D \cdot z_d & = \sum_{j \in [D']} z_{d,j}. \nonumber
\end{align}

We prove by induction on $d$ that it implies $D'=D$ and $\bv_{i,j}=\bv_i$ for all $i \in [d]$ and $j \in [D]$.
The base case is $d=1$, where we have
\[
D \cdot \bv_1 = \sum_{j \in [D']} \bv_{1,j}.
\]
The claim then follows from strict convexity of $\cV(r)$ as $D' \le D$.
Now consider $d>1$.
Since
\[
D \cdot z_d = \sum_{j \in [D']} z_{d,j},
\]
again by strict convexity of $\cV(r)$, we can get
\begin{align}
D \cdot y_d \ge \sum_{j \in [D']} y_{d,j}. \label{eqn:unique-con}
\end{align}
If \cref{eqn:unique-con} is actually an equality, by \cref{eqn:unique-sum}, we also have
\[
D \cdot z_{d-1} = \sum_{j \in [D']} z_{d-1,j}.
\]
Thus, the claim follows from the inductive hypothesis for $d-1$ and the base case.
Otherwise, \cref{eqn:unique-con} is a strict inequality, implying that
\[
D \cdot z_{d-1} < \sum_{j \in [D']} z_{d-1,j},
\]
due to \cref{eqn:unique-sum}.
By repeatedly applying the strict convexity argument, we would finally conclude
\[
D \cdot y_1 > \sum_{j \in [D']} y_{1,j},
\]
a contradiction.

Now we show \cref{item:hs-P-overlap}.
Fix $P,P' \in \cP$.
By \cref{alg:path}, $P \cap P' = \emptyset$ if they use the same $d$ direction vectors, as they would be added to $\cP$ in the same iteration of the for loop.
Otherwise, they must use different sets of $d$ direction vectors.
Observe that \cref{item:hs-P-unique} also implies $P \cap P'$ can only be a consecutive subpath of both $P,P'$.
Meanwhile, any path of $d$ edges uniquely determines all $d$ direction vectors.
Altogether, it must be that $|P \cap P'| < d$.

To see the size of $\cP$, it is sufficient to show that each iteration of the for loop in \cref{alg:path} adds $\Theta(|V|/(dD))$ critical paths.
Fix an iteration for $\bv_1,\dots,\bv_d$, which uses only the edges from $(i,\bu)$ to $((i \bmod d)+1,\bu+\be_i(\bv_i))$ for $i \in [d]$ and $\bu \in [4Dr]^{d+1}$.
Note that each vertex $(i,\bu)$ with $i \in [d]$ and $\bu \in [2Dr]^{d+1}$ has one such edge.
Let $E''$ be the set of these $|V|/2^{d+1}$ edges and $\cP'$ the set of critical paths added in this iteration.
We claim that $|\cP'| = \Theta(|E''|/(dD)) = \Theta(|V|/(dD))$, as desired.

Regarding the claim, consider an unused edge $e \in E''$ from $(i',\bu)$ to $((i' \bmod d)+1,\bu+\be_{i'}(\bv_{i'}))$ for some $i' \in [d]$ and $\bu \in [2Dr]^{d+1}$.
Let $P'$ be the unique path from $(i',\bu)$ with $dD$ edges, by following the directions $\bv_i$ alternately.
$P'$ is well defined as $\bu \in [2Dr]^{d+1}$ and each of its coordinates increases by at most $2Dr$.
Since $e$ is unused, there exists $P \in \cP'$ such that $P \cap P' \ne \emptyset$.
Observe that this happens at most $dD$ times for any fixed $P$, one for each of $dD$ ``backtracing'' steps in the directions $-\bv_i$ alternately.
So each such unused edge $e \in E''$ can be paired with a unique edge of $P$.
In turn, this means $E''$ has at most $|\cP'| \cdot dD$ unused edges and at most $|\cP'| \cdot dD$ used edges.
Combined with $|E''|=|V|/2^{d+1}$, we can get
\[
2|\cP'| \cdot dD \ge |V|/2^{d+1},
\]
concluding the proof.

\paragraph{From directed to undirected.}

In fact, almost identical arguments still work if the directions of the edges are eliminated, except that we now need to show that each critical path $P \in \cP$ found by \cref{alg:path} remains the unique shortest path between its endpoints.
This is necessary because each edge can now be traversed in both directions.
We show that such uniqueness indeed holds for $d \in \set{1,2}$, which is sufficient for our purposes, in the rest of this section.

The case of $d=1$ follows straightforwardly from strict convexity.
To see this, each critical path $P \in \cP$ must has the form $s=(1,\bu),(1,\bu+\bv),\dots,t=(1,\bu+D \cdot \bv)$, for some $\bu \in [4Dr]^2$ and $\bv \in \cV(r)$.
Due to strict convexity, this is the unique $s$-$t$ path with at most $D$ edges, each of which corresponds to a vector in $\cV(r) \cup (-\cV(r)) \subseteq \cB(r)$.

For $d=2$, fix a critical path $P$ between $s=(i,\bu)$ and $t=(i,\bu')$ for some $i \in [2]$ and $\bu,\bu' \in [4Dr]^3$.
Without loss of generality, assume $i=1$.
The case of $i=2$ is symmetric.
Let $\bv_1=(y_1,z_1),\bv_2=(y_2,z_2) \in \cV(r)$ be the two direction vectors induced by $P$.
Let $P'$ be any path between $s$ and $t'=(1,\bu'')$ for some $\bu'' \in [4Dr]^3$, with at most $2D$ edges.
We claim that $P$ is the unique such path maximizing $\langle \bu''-\bu, \bv' \rangle$, where $\bv'=(y_1/z_1,1,z_2/y_2)$.
Observe that our constructed graph is bipartite for $d=2$.
As a result, each odd step of $P'$ either traverses an edge corresponding to $\be_1(\bv'_1)$ for some $\bv'_1 \in \cV(r)$ in the forward direction, or traverses an edge corresponding to $\be_2(\bv'_2)$ for some $\bv'_2 \in \cV(r)$ in the backward direction.
In the former case, $\langle \be_1(\bv'_1), \bv' \rangle > 0$ and $\bv'_1=\bv_1$ is the unique maximizer due to strict convexity.
In the latter case, we always have $\langle -\be_2(\bv'_2), \bv' \rangle < 0$ because $\bv'_2 \in \cV(r)$ can only have positive coordinates by definition.
Thus, $P'$ can maximize $\langle \bu''-\bu, \bv' \rangle$ only if it traverses the edge corresponding to $\be_1(\bv_1)$ in the forward direction, as does $P$.
A similar analysis holds for even steps.
This concludes the proof of the claim.
Altogether, $P$ must be the unique $s$-$t$ shortest path.
\end{proof}

\subsection{Proof of \cref{lem:rp}}
\label{sec:proof-rp}

\lemrp*

\begin{proof}
The starting point is a construction of \cite{bodwin2021new}, originally introduced in the context of reachability and distance preservers.
At a high level, the graph is obtained via obstacle products, where both the inner and the outer graphs use constructions of \cite{HuangP21}, summarized in \cref{lem:hp}.

The outer graph has $3$ layers of high dimensional grids.
Regarding the inner graph, for reachability and distance preservers as in \cite{bodwin2021new}, it suffices to use $2$-dimensional grids in each layer because the lower bounds are typically concerned with $\cO(n)$-size preservers, as preservers of size $\cO(m)$ make little sense.
For us, to obtain $m^{1+\Omega(1)}$ lower bounds on certification complexity, we have to first apply alternation products to increase the number of critical paths.
This is also by now the standard way of extending diameter lower bounds on $\cO(n)$-size shortcuts to $\cO(m)$-size shortcuts.
Specifically, we use the improved alternation product of \cite{LuWWX22} as given in \cref{lem:lvwx}.

Readers are referred to for example \cite{CoppersmithE06,AbboudB24,AbboudB17,BodwinH22,BodwinHWWX24} for any necessary background knowledge and more detailed discussion on reachability and distance preservers, as well as alternation and obstacle products.

In the rest of this section, we use subscripts $\sO,\sI$ to distinguish between variables for the outer and the inner graphs.
Concretely, set $d_\sO=\sqrt{\log n_\sO}$ and $D_\sO=3$ for the outer graph $G_\sO$.
By \cref{lem:hp}, we can get $r_\sO=\Theta(n_\sO^{1/d_\sO}/D_\sI^{(d_\sO+1)/d_\sO})=2^{d_\sO-o(1)}$ and $\Delta_\sO=\Theta(r_\sO^{d_\sO(d_\sO-1)/(d_\sO+1)})=n_\sO^{1-o(1)}$.
Moreover, $G_\sO$ has $|\cP_\sO|=\Theta(n_\sO \Delta_\sO / D_\sO)=n_\sO^{2-o(1)}$ edge-disjoint critical paths.

In the meantime, set $d_\sI=2$ and $r_\sI=n_\sI^\epsilon$ for the inner graph $G_\sI$.
By \cref{lem:lvwx}, we also get $D_\sI=\Theta((n_\sI/d_\sI)^{1/(d_\sI+2)}/r_\sI^{(d_\sI+1)/(d_\sI+2)})=\Theta(n_\sI^{1/4-3\epsilon/4})$ and $\Delta_\sI=\Theta(r_\sI^{2/3})=\Theta(n_\sI^{2\epsilon/3})$.
Moreover, $G_\sI$ has $|\cP_\sI|=\Theta(n_\sI \Delta_\sI^{d_\sI}/(d_\sI D_\sI))=\Theta(n_\sI^{3/4+25\epsilon/12})$ critical paths, any two of which share at most one edge.

Now, we are ready to apply obstacle products.
In particular, the final graph $G$ is obtained by replacing each vertex in the middle layer of $G_\sO$ with a disjoint copy of $G_\sI$, according to the following procedure.
Fix a vertex $v$ in the middle layer of $G_\sO$.
We write $G^{(v)}$ to denote the copy of $G_\sI$ corresponding to $v$, and use $u^{(v)}$ to identify the copy of vertex $u$ in $G^{(v)}$.
Since critical paths in $G_\sO$ are edge-disjoint, at most $\Delta_\sO$ of them contain $v$.
For every such path $(x,v,y)$, arbitrarily pick a distinct critical path in $G_\sI$. 
Let $s,t$ be its endpoints.
By adding the edges $(x,s^{(v)})$ and $(t^{(v)},y)$, there is now a path from $x$ to $y$ in $G$ by following the critical path from $s^{(v)}$ to $t^{(v)}$ in $G^{(v)}$.
Note that the above procedure is well-defined so long as $|\cP_\sI| \ge \Delta_\sO$.
Let $\cP$ be the set of all paths obtained by this procedure.

We show that \cref{lem:rp} is satisfied by $G$ and $\cP$.
To this end, observe that the total number of vertices in $G$ is $n=\Theta(n_\sO n_\sI)$.
By setting $\Delta_\sO=|\cP_\sI|$, we can get $n_\sO=n^{3/7+100\epsilon/147-o(\epsilon)}$ and $n_\sI=n^{4/7-100\epsilon/147+o(\epsilon)}$.
As both $G_\sO$ and $G_\sI$ are layered DAGs, $G$ is also a layered DAG due to obstacle products, and it has $D=\Theta(D_\sI)=\Theta(n_\sI^{1/4-3\epsilon/4})=n^{1/7-\Theta(\epsilon)}$ layers, proving \cref{item:rp-D}.
For \cref{item:rp-S}, the first and the last layers each have $\Theta(n_\sO)=n^{3/7+100\epsilon/147-o(\epsilon)} \le n^{1/2}$ vertices for $\epsilon \in [0,21/200]$.

As to \cref{item:rp-P}, \cref{item:rp-P-length} is straightforward by construction.
Also recall that each critical path in $G$ is obtained from a critical path $(x,v,y)$ in $G_\sO$, by replacing $v$ with some critical path from $s^{(v)}$ to $t^{(v)}$ in $G^{(v)}$.
Since $(x,v,y)$ is the unique path from $x$ to $y$ in $G_\sO$, any path from $x$ to $y$ in $G$ can only use vertices in $G^{(v)}$.
Moreover, since critical paths in $G_\sO$ are edge-disjoint, $(x,s^{(v)})$ is the only edge in $G$ from $x$ to any vertex in $G^{(v)}$.
Similarly, $(t^{(v)},y)$ is the only edge in $G$ from any vertex in $G^{(v)}$ to $y$.
As a result, the uniqueness of $s$-$t$ path in $G_\sI$ implies the uniqueness of $x$-$y$ path in $G$, hence proving \cref{item:rp-P-unique}.
To see \cref{item:rp-P-overlap}, consider any two critical paths $P,P'$ in $G$ obtained from critical paths $(x,v,y),(x',v',y')$ in $G_\sO$ respectively, which are edge-disjoint by \cref{item:hp-P-overlap} of \cref{lem:hp}.
If they are actually vertex-disjoint in $G_\sO$, then $P,P'$ are also vertex-disjoint in $G$.
If $x=x'$ or $y=y'$, we must have $v \ne v'$, so $P,P'$ share a single vertex, namely either $x$ or $y$, in $G$ as well.
Otherwise, $v=v'$.
In such case, $P,P'$ must follow two distinct critical paths in $G^{(v)}$.
Thus, their intersection is at most one edge by \cref{item:lvwx-P-overlap} of \cref{lem:lvwx}.

Finally, we calculate the total number of edges and of critical paths in $G$.
By construction, it has $m=\Theta(m_\sO+n_\sO m_\sI)=\Theta(n_\sO n_\sI \Delta_\sI)=\Theta(n\Delta_\sI)$.
Meanwhile, we have $|\cP|=|\cP_\sO|=\Theta(n_\sO \Delta_\sO)=\Theta(n_\sO |\cP_\sI|)=\Theta(n_\sO n_\sI \Delta_\sI^2 / D_\sI)=\Theta(n \Delta_\sI^2 / D)$.
Note that $\Delta_\sI=\Theta(n_\sI^{2\epsilon/3})=n^{8\epsilon/21-o(\epsilon)}$.
Therefore, adjusting $\epsilon$ by an appropriate constant factor in the construction concludes the proof.
\end{proof}

%% file: appendix_treaps.tex
\section{Treaps}\label{sec:treaps}

A treap \cite{AragonS89} is a randomized balanced binary tree data structure representing an ordered sequence $S = s_1, \dots, s_n$ of elements, each of which is associated with a randomly sampled \textit{priority} $p(s_i)$. The layout of the treap is entirely determined by the sequence $S$ and the priorities: the root of the treap is the element $s_j$ of minimum priority $p(s_j)$, its left child is the root of a treap built on elements $s_1, \dots, s_{j - 1}$, and its right child is the root of a treap built on elements $s_{j + 1}, \dots, s_n$ (using the same global priorities $p(e_i)$). Each vertex also stores the size of its subtree.

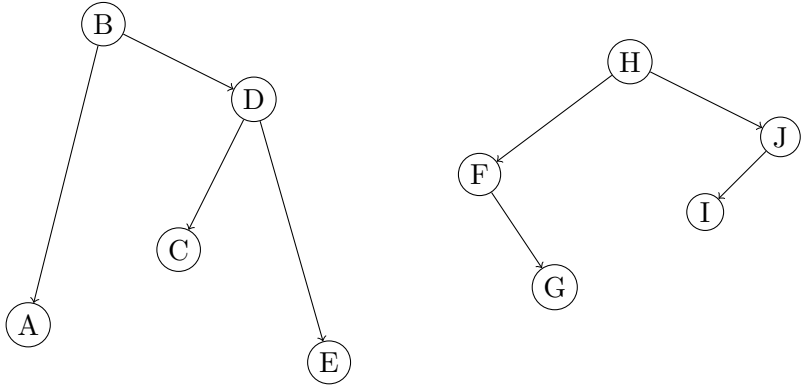
\begin{figure}[tp]
    \centering
    \input{fig/treap1}
    \caption{Two treaps on sequences ABCDE and FGHIJ respectively, with the priorities represented by y-coordinates, vertices of smaller priority being higher up.}
    \label{fig:treap1}
\end{figure}

Note that the depth of a treap on an $n$-length sequence is $O(\log n)$ with high probability: the minimum-priority element is in the middle half of elements with probability $\frac{1}{2}$, in which case the sequences for both the left and right subtrees are a constant fraction smaller. Furthermore, given fixed sampled priorities, a treap has the same layout and in particular depth regardless of the operations that produced it. Thus, for a sequence of $\poly(n)$ operations on a number of treaps of total size $n$, each treap produced has depth $O(\log n)$ with high probability.

A treap supports two main operations:
\begin{itemize}
    \item \textsc{Join}: given two treaps with sequences $S_1$ and $S_2$, construct a treap for the sequence $S = S_1 S_2$. This operation can be performed in time proportional to the depth of the resulting treap, which is $O(\log |S|)$ with high probability.
    \item \textsc{Split}: given a treap and a split point $k$, split the treap into two treaps for the prefix $S_1$ and suffix $S_2$ of the sequence $S$, where $|S_1| = k$ and $|S_2| = |S| - k$. This operation can be performed in time proportional to the depth of the input treap, which is $O(\log |S|)$ with high probability.
\end{itemize}
Both of these operations are simple to implement recursively. Specifically, the \textsc{join} operation on two treaps can be performed as follows:
\begin{enumerate}
    \item Set the root of the result to be the lower-priority root of the two input treaps. Suppose this is the root of the left treap.
    \item Recursively \textsc{join} the right subtree of the left treap's root and the right treap, to form the treap for the subsequence right of the root.
    \item Set the right subtree of the root to be the result of this merge.
\end{enumerate}
The \textsc{split} operation on a treap and a split point $k$ can be performed as follows:
\begin{enumerate}
    \item If the left subtree has size at least $k$, recursively split it, set the right tree of the recursive split to be the new left subtree, and return the left tree of the recursive split and the old root.
    \item Otherwise, if the left subtree has size strictly less than $k$, recursively split the right subtree at $k - \text{size of the left subtree} - 1$. Set the left tree of the recursive split to be the new right subtree, and return the old root and the right tree of the recursive split.
\end{enumerate}
The subtree sizes can be maintained by updating the size whenever a child is changed or recursion returns to a vertex.

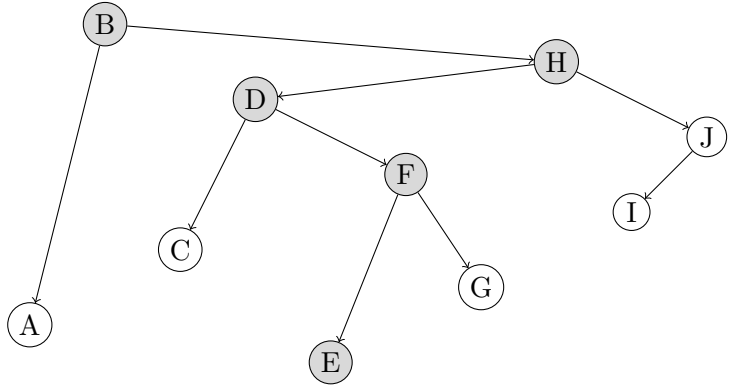
\begin{figure}[tp]
    \centering
    \input{fig/treap2}
    \caption{The treap resulting from a \textsc{join} operation on the two treaps of \Cref{fig:treap1}. The vertices visited by the operation are highlighted in gray. Subtrees of unvisited vertices remain identical.}
    \label{fig:treap2}
\end{figure}

\begin{figure}[tp]
    \centering
    \input{fig/treap3}
    \caption{The two treaps resulting from a \textsc{split} operation on the treap of \Cref{fig:treap2} with $k = 6$. The vertices visited by the operation are highlighted in gray. Subtrees of unvisited vertices remain identical.}
    \label{fig:treap3}
\end{figure}
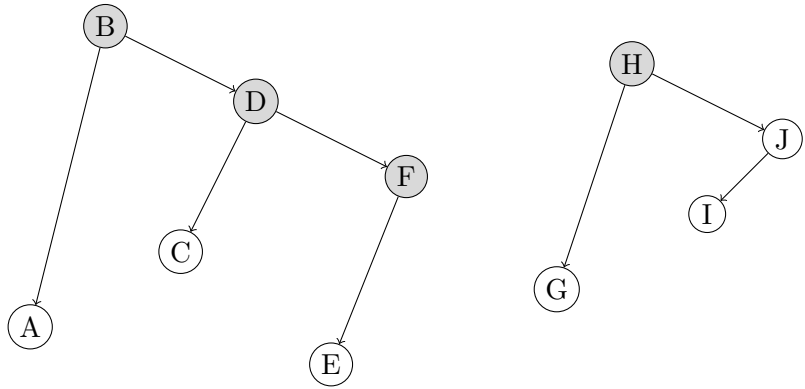

%% file: fig/treap1.tex
\begin{tikzpicture}
    \node[draw, circle,inner sep=2pt] (A) at (0, -4) {A};
    \node[draw, circle,inner sep=2pt] (B) at (1, 0) {B};
    \node[draw, circle,inner sep=2pt] (C) at (2, -3) {C};
    \node[draw, circle,inner sep=2pt] (D) at (3, -1) {D};
    \node[draw, circle,inner sep=2pt] (E) at (4, -4.5) {E};
    \draw[->] (B) edge (A);
    \draw[->] (B) edge (D);
    \draw[->] (D) edge (C);
    \draw[->] (D) edge (E);

    \node[draw, circle,inner sep=2pt] (F) at (6, -2) {F};
    \node[draw, circle,inner sep=2pt] (G) at (7, -3.5) {G};
    \node[draw, circle,inner sep=2pt] (H) at (8, -0.5) {H};
    \node[draw, circle,inner sep=2pt] (I) at (9, -2.5) {I};
    \node[draw, circle,inner sep=2pt] (J) at (10, -1.5) {J};
    \draw[->] (H) edge (F);
    \draw[->] (H) edge (J);
    \draw[->] (J) edge (I);
    \draw[->] (F) edge (G);
\end{tikzpicture}

%% file: fig/treap2.tex
\begin{tikzpicture}

    \node[draw, circle,inner sep=2pt] (A) at (0, -4) {A};
    \node[draw, circle,inner sep=2pt, fill=black!15] (B) at (1, 0) {B};
    \node[draw, circle,inner sep=2pt] (C) at (2, -3) {C};
    \node[draw, circle,inner sep=2pt, fill=black!15] (D) at (3, -1) {D};
    \node[draw, circle,inner sep=2pt, fill=black!15] (E) at (4, -4.5) {E};

    \node[draw, circle,inner sep=2pt, fill=black!15] (F) at (5, -2) {F};
    \node[draw, circle,inner sep=2pt] (G) at (6, -3.5) {G};
    \node[draw, circle,inner sep=2pt, fill=black!15] (H) at (7, -0.5) {H};
    \node[draw, circle,inner sep=2pt] (I) at (8, -2.5) {I};
    \node[draw, circle,inner sep=2pt] (J) at (9, -1.5) {J};
    
    \draw[->] (B) edge (A);
    \draw[->] (B) edge (H);
    \draw[->] (D) edge (C);
    \draw[->] (D) edge (F);
    \draw[->] (H) edge (D);
    \draw[->] (H) edge (J);
    \draw[->] (J) edge (I);
    \draw[->] (F) edge (E);
    \draw[->] (F) edge (G);
\end{tikzpicture}

%% file: fig/treap3.tex
\begin{tikzpicture}
    \node[draw, circle,inner sep=2pt] (A) at (0, -4) {A};
    \node[draw, circle,inner sep=2pt, fill=black!15] (B) at (1, 0) {B};
    \node[draw, circle,inner sep=2pt] (C) at (2, -3) {C};
    \node[draw, circle,inner sep=2pt, fill=black!15] (D) at (3, -1) {D};
    \node[draw, circle,inner sep=2pt] (E) at (4, -4.5) {E};

    \node[draw, circle,inner sep=2pt, fill=black!15] (F) at (5, -2) {F};
    \node[draw, circle,inner sep=2pt] (G) at (7, -3.5) {G};
    \node[draw, circle,inner sep=2pt, fill=black!15] (H) at (8, -0.5) {H};
    \node[draw, circle,inner sep=2pt] (I) at (9, -2.5) {I};
    \node[draw, circle,inner sep=2pt] (J) at (10, -1.5) {J};
    
    \draw[->] (B) edge (A);
    \draw[->] (B) edge (D);
    \draw[->] (D) edge (C);
    \draw[->] (D) edge (F);
    \draw[->] (F) edge (E);
    
    \draw[->] (H) edge (G);
    \draw[->] (H) edge (J);
    \draw[->] (J) edge (I);
\end{tikzpicture}

%% file: main.bib
@incollection{Raskhodnikova10,
  author       = {Sofya Raskhodnikova},
  editor       = {Oded Goldreich},
  title        = {Transitive-Closure Spanners: {A} Survey},
  booktitle    = {Property Testing - Current Research and Surveys},
  series       = {Lecture Notes in Computer Science},
  volume       = {6390},
  pages        = {167--196},
  publisher    = {Springer},
  year         = {2010},
  url          = {https://doi.org/10.1007/978-3-642-16367-8\_10},
  doi          = {10.1007/978-3-642-16367-8\_10},
  bibsource    = {dblp computer science bibliography, https://dblp.org}
}

@inproceedings{BodwinH22,
  author       = {Greg Bodwin and
                  Gary Hoppenworth},
  title        = {New Additive Spanner Lower Bounds by an Unlayered Obstacle Product},
  booktitle    = {63rd {IEEE} Annual Symposium on Foundations of Computer Science, {FOCS}
                  2022, Denver, CO, USA, October 31 - November 3, 2022},
  pages        = {778--788},
  publisher    = {{IEEE}},
  year         = {2022},
  url          = {https://doi.org/10.1109/FOCS54457.2022.00079},
  doi          = {10.1109/FOCS54457.2022.00079},
  bibsource    = {dblp computer science bibliography, https://dblp.org}
}

@article{AbboudB24,
  author       = {Amir Abboud and
                  Greg Bodwin},
  title        = {Reachability Preservers: New Extremal Bounds and Approximation Algorithms},
  journal      = {{SIAM} J. Comput.},
  volume       = {53},
  number       = {2},
  pages        = {221--246},
  year         = {2024},
  url          = {https://doi.org/10.1137/21m1442176},
  doi          = {10.1137/21M1442176},
  bibsource    = {dblp computer science bibliography, https://dblp.org}
}

@article{AbboudB17,
  author       = {Amir Abboud and
                  Greg Bodwin},
  title        = {The 4/3 Additive Spanner Exponent Is Tight},
  journal      = {J. {ACM}},
  volume       = {64},
  number       = {4},
  pages        = {28:1--28:20},
  year         = {2017},
  url          = {https://doi.org/10.1145/3088511},
  doi          = {10.1145/3088511},
  bibsource    = {dblp computer science bibliography, https://dblp.org}
}

@article{bals2025greedy,
  title={Greedy Algorithms for Shortcut Sets and Hopsets},
  author={Bals, Ben and Blikstad, Joakim and Bodwin, Greg and Dadush, Daniel and Forster, Sebastian and Nazari, Yasamin},
  journal={arXiv preprint arXiv:2511.20111},
  year={2025}
}

@article{bodwin2021new,
  title={New results on linear size distance preservers},
  author={Bodwin, Greg},
  journal={SIAM Journal on Computing},
  volume={50},
  number={2},
  pages={662--673},
  year={2021},
  publisher={SIAM}
}

@article{CoppersmithE06,
  author       = {Don Coppersmith and
                  Michael Elkin},
  title        = {Sparse Sourcewise and Pairwise Distance Preservers},
  journal      = {{SIAM} J. Discret. Math.},
  volume       = {20},
  number       = {2},
  pages        = {463--501},
  year         = {2006},
  url          = {https://doi.org/10.1137/050630696},
  doi          = {10.1137/050630696},
  bibsource    = {dblp computer science bibliography, https://dblp.org}
}

@inproceedings{Thorup92,
  author       = {Mikkel Thorup},
  editor       = {Ernst W. Mayr},
  title        = {On Shortcutting Digraphs},
  booktitle    = {Graph-Theoretic Concepts in Computer Science, 18th International Workshop,
                  {WG} '92, Wiesbaden-Naurod, Germany, June 19-20, 1992, Proceedings},
  series       = {Lecture Notes in Computer Science},
  volume       = {657},
  pages        = {205--211},
  publisher    = {Springer},
  year         = {1992},
  url          = {https://doi.org/10.1007/3-540-56402-0\_48},
  doi          = {10.1007/3-540-56402-0\_48},
  bibsource    = {dblp computer science bibliography, https://dblp.org}
}

@inproceedings{Hesse03,
  author       = {William Hesse},
  title        = {Directed graphs requiring large numbers of shortcuts},
  booktitle    = {Proceedings of the Fourteenth Annual {ACM-SIAM} Symposium on Discrete
                  Algorithms, January 12-14, 2003, Baltimore, Maryland, {USA}},
  pages        = {665--669},
  publisher    = {{ACM/SIAM}},
  year         = {2003},
  url          = {http://dl.acm.org/citation.cfm?id=644108.644216},
  bibsource    = {dblp computer science bibliography, https://dblp.org}
}

@inproceedings{KoganP23,
  author       = {Shimon Kogan and
                  Merav Parter},
  editor       = {Nikhil Bansal and
                  Viswanath Nagarajan},
  title        = {Faster and Unified Algorithms for Diameter Reducing Shortcuts and
                  Minimum Chain Covers},
  booktitle    = {Proceedings of the 2023 {ACM-SIAM} Symposium on Discrete Algorithms,
                  {SODA} 2023, Florence, Italy, January 22-25, 2023},
  pages        = {212--239},
  publisher    = {{SIAM}},
  year         = {2023},
  url          = {https://doi.org/10.1137/1.9781611977554.ch9},
  doi          = {10.1137/1.9781611977554.CH9},
  bibsource    = {dblp computer science bibliography, https://dblp.org}
}

@article{barany1998convex,
  title={The convex hull of the integer points in a large ball},
  author={B{\'a}r{\'a}ny, Imre and Larman, David G},
  journal={Mathematische Annalen},
  volume={312},
  pages={167--182},
  year={1998},
  publisher={Springer-Verlag}
}

@article{UllmanY91,
  author       = {Jeffrey D. Ullman and
                  Mihalis Yannakakis},
  title        = {High-Probability Parallel Transitive-Closure Algorithms},
  journal      = {{SIAM} J. Comput.},
  volume       = {20},
  number       = {1},
  pages        = {100--125},
  year         = {1991},
  url          = {https://doi.org/10.1137/0220006},
  doi          = {10.1137/0220006},
  bibsource    = {dblp computer science bibliography, https://dblp.org}
}

@inproceedings{JambulapatiLS19,
  author       = {Arun Jambulapati and
                  Yang P. Liu and
                  Aaron Sidford},
  editor       = {David Zuckerman},
  title        = {Parallel Reachability in Almost Linear Work and Square Root Depth},
  booktitle    = {60th {IEEE} Annual Symposium on Foundations of Computer Science, {FOCS}
                  2019, Baltimore, Maryland, USA, November 9-12, 2019},
  pages        = {1664--1686},
  publisher    = {{IEEE} Computer Society},
  year         = {2019},
  url          = {https://doi.org/10.1109/FOCS.2019.00098},
  doi          = {10.1109/FOCS.2019.00098},
  bibsource    = {dblp computer science bibliography, https://dblp.org}
}

@inproceedings{BermanRR10,
  author       = {Piotr Berman and
                  Sofya Raskhodnikova and
                  Ge Ruan},
  editor       = {Kamal Lodaya and
                  Meena Mahajan},
  title        = {Finding Sparser Directed Spanners},
  booktitle    = {{IARCS} Annual Conference on Foundations of Software Technology and
                  Theoretical Computer Science, {FSTTCS} 2010, December 15-18, 2010,
                  Chennai, India},
  series       = {LIPIcs},
  volume       = {8},
  pages        = {424--435},
  publisher    = {Schloss Dagstuhl - Leibniz-Zentrum f{\"{u}}r Informatik},
  year         = {2010},
  url          = {https://doi.org/10.4230/LIPIcs.FSTTCS.2010.424},
  doi          = {10.4230/LIPICS.FSTTCS.2010.424},
  bibsource    = {dblp computer science bibliography, https://dblp.org}
}

@inproceedings{KoganP22,
  author       = {Shimon Kogan and
                  Merav Parter},
  editor       = {Joseph (Seffi) Naor and
                  Niv Buchbinder},
  title        = {New Diameter-Reducing Shortcuts and Directed Hopsets: Breaking the
                  Barrier},
  booktitle    = {Proceedings of the 2022 {ACM-SIAM} Symposium on Discrete Algorithms,
                  {SODA} 2022, Virtual Conference / Alexandria, VA, USA, January 9 -
                  12, 2022},
  pages        = {1326--1341},
  publisher    = {{SIAM}},
  year         = {2022},
  url          = {https://doi.org/10.1137/1.9781611977073.55},
  doi          = {10.1137/1.9781611977073.55},
  bibsource    = {dblp computer science bibliography, https://dblp.org}
}

@inproceedings{BodwinH23,
  author       = {Greg Bodwin and
                  Gary Hoppenworth},
  title        = {Folklore Sampling is Optimal for Exact Hopsets: Confirming the {\(\surd\)}n
                  Barrier},
  booktitle    = {64th {IEEE} Annual Symposium on Foundations of Computer Science, {FOCS}
                  2023, Santa Cruz, CA, USA, November 6-9, 2023},
  pages        = {701--720},
  publisher    = {{IEEE}},
  year         = {2023},
  url          = {https://doi.org/10.1109/FOCS57990.2023.00046},
  doi          = {10.1109/FOCS57990.2023.00046},
  bibsource    = {dblp computer science bibliography, https://dblp.org}
}

@inproceedings{HoppenworthXX25,
  author       = {Gary Hoppenworth and
                  Yinzhan Xu and
                  Zixuan Xu},
  editor       = {Yossi Azar and
                  Debmalya Panigrahi},
  title        = {New Separations and Reductions for Directed Hopsets and Preservers},
  booktitle    = {Proceedings of the 2025 Annual {ACM-SIAM} Symposium on Discrete Algorithms,
                  {SODA} 2025, New Orleans, LA, USA, January 12-15, 2025},
  pages        = {4405--4443},
  publisher    = {{SIAM}},
  year         = {2025},
  url          = {https://doi.org/10.1137/1.9781611978322.150},
  doi          = {10.1137/1.9781611978322.150},
  bibsource    = {dblp computer science bibliography, https://dblp.org}
}

@inproceedings{LuWWX22,
  author       = {Kevin Lu and
                  Virginia {Vassilevska Williams} and
                  Nicole Wein and
                  Zixuan Xu},
  editor       = {Joseph (Seffi) Naor and
                  Niv Buchbinder},
  title        = {Better Lower Bounds for Shortcut Sets and Additive Spanners via an
                  Improved Alternation Product},
  booktitle    = {Proceedings of the 2022 {ACM-SIAM} Symposium on Discrete Algorithms,
                  {SODA} 2022, Virtual Conference / Alexandria, VA, USA, January 9 -
                  12, 2022},
  pages        = {3311--3331},
  publisher    = {{SIAM}},
  year         = {2022},
  url          = {https://doi.org/10.1137/1.9781611977073.131},
  doi          = {10.1137/1.9781611977073.131},
  bibsource    = {dblp computer science bibliography, https://dblp.org}
}

@inproceedings{BodwinHWWX24,
  author       = {Greg Bodwin and
                  Gary Hoppenworth and
                  Virginia {Vassilevska Williams} and
                  Nicole Wein and
                  Zixuan Xu},
  editor       = {Karl Bringmann and
                  Martin Grohe and
                  Gabriele Puppis and
                  Ola Svensson},
  title        = {Additive Spanner Lower Bounds with Optimal Inner Graph Structure},
  booktitle    = {51st International Colloquium on Automata, Languages, and Programming,
                  {ICALP} 2024, July 8-12, 2024, Tallinn, Estonia},
  series       = {LIPIcs},
  volume       = {297},
  pages        = {28:1--28:17},
  publisher    = {Schloss Dagstuhl - Leibniz-Zentrum f{\"{u}}r Informatik},
  year         = {2024},
  url          = {https://doi.org/10.4230/LIPIcs.ICALP.2024.28},
  doi          = {10.4230/LIPICS.ICALP.2024.28},
  bibsource    = {dblp computer science bibliography, https://dblp.org}
}

@inproceedings{WilliamsXX24,
  author       = {Virginia {Vassilevska Williams} and
                  Yinzhan Xu and
                  Zixuan Xu},
  editor       = {David P. Woodruff},
  title        = {Simpler and Higher Lower Bounds for Shortcut Sets},
  booktitle    = {Proceedings of the 2024 {ACM-SIAM} Symposium on Discrete Algorithms,
                  {SODA} 2024, Alexandria, VA, USA, January 7-10, 2024},
  pages        = {2643--2656},
  publisher    = {{SIAM}},
  year         = {2024},
  url          = {https://doi.org/10.1137/1.9781611977912.94},
  doi          = {10.1137/1.9781611977912.94},
  bibsource    = {dblp computer science bibliography, https://dblp.org}
}

@article{HuangP21,
  author       = {Shang{-}En Huang and
                  Seth Pettie},
  title        = {Lower Bounds on Sparse Spanners, Emulators, and Diameter-Reducing
                  Shortcuts},
  journal      = {{SIAM} J. Discret. Math.},
  volume       = {35},
  number       = {3},
  pages        = {2129--2144},
  year         = {2021},
  url          = {https://doi.org/10.1137/19M1306154},
  doi          = {10.1137/19M1306154},
  bibsource    = {dblp computer science bibliography, https://dblp.org}
}

@inproceedings{AragonS89,
  author       = {Cecilia R. Aragon and
                  Raimund Seidel},
  title        = {Randomized Search Trees},
  booktitle    = {30th Annual Symposium on Foundations of Computer Science, Research
                  Triangle Park, North Carolina, USA, 30 October - 1 November 1989},
  pages        = {540--545},
  publisher    = {{IEEE} Computer Society},
  year         = {1989},
  url          = {https://doi.org/10.1109/SFCS.1989.63531},
  doi          = {10.1109/SFCS.1989.63531},
  bibsource    = {dblp computer science bibliography, https://dblp.org}
}

@inproceedings{Fineman18,
  author       = {Jeremy T. Fineman},
  editor       = {Ilias Diakonikolas and
                  David Kempe and
                  Monika Henzinger},
  title        = {Nearly work-efficient parallel algorithm for digraph reachability},
  booktitle    = {Proceedings of the 50th Annual {ACM} {SIGACT} Symposium on Theory
                  of Computing, {STOC} 2018, Los Angeles, CA, USA, June 25-29, 2018},
  pages        = {457--470},
  publisher    = {{ACM}},
  year         = {2018},
  url          = {https://doi.org/10.1145/3188745.3188926},
  doi          = {10.1145/3188745.3188926},
  bibsource    = {dblp computer science bibliography, https://dblp.org}
}

@inproceedings{Caceres23,
  author       = {Manuel C{\'{a}}ceres},
  editor       = {Kousha Etessami and
                  Uriel Feige and
                  Gabriele Puppis},
  title        = {Minimum Chain Cover in Almost Linear Time},
  booktitle    = {50th International Colloquium on Automata, Languages, and Programming,
                  {ICALP} 2023, Paderborn, Germany, July 10-14, 2023},
  series       = {LIPIcs},
  volume       = {261},
  pages        = {31:1--31:12},
  publisher    = {Schloss Dagstuhl - Leibniz-Zentrum f{\"{u}}r Informatik},
  year         = {2023},
  url          = {https://doi.org/10.4230/LIPIcs.ICALP.2023.31},
  doi          = {10.4230/LIPICS.ICALP.2023.31},
  bibsource    = {dblp computer science bibliography, https://dblp.org}
}

@inproceedings{KoganP22fast,
  author       = {Shimon Kogan and
                  Merav Parter},
  editor       = {Mikolaj Bojanczyk and
                  Emanuela Merelli and
                  David P. Woodruff},
  title        = {Beating Matrix Multiplication for n{\^{}}\{1/3\}-Directed Shortcuts},
  booktitle    = {49th International Colloquium on Automata, Languages, and Programming,
                  {ICALP} 2022, Paris, France, July 4-8, 2022},
  series       = {LIPIcs},
  volume       = {229},
  pages        = {82:1--82:20},
  publisher    = {Schloss Dagstuhl - Leibniz-Zentrum f{\"{u}}r Informatik},
  year         = {2022},
  url          = {https://doi.org/10.4230/LIPIcs.ICALP.2022.82},
  doi          = {10.4230/LIPICS.ICALP.2022.82},
  bibsource    = {dblp computer science bibliography, https://dblp.org}
}

@inproceedings{Brand0PKLGSS23flow,
  author       = {Jan van den Brand and
                  Li Chen and
                  Richard Peng and
                  Rasmus Kyng and
                  Yang P. Liu and
                  Maximilian Probst Gutenberg and
                  Sushant Sachdeva and
                  Aaron Sidford},
  title        = {A Deterministic Almost-Linear Time Algorithm for Minimum-Cost Flow},
  booktitle    = {64th {IEEE} Annual Symposium on Foundations of Computer Science, {FOCS}
                  2023, Santa Cruz, CA, USA, November 6-9, 2023},
  pages        = {503--514},
  publisher    = {{IEEE}},
  year         = {2023},
  url          = {https://doi.org/10.1109/FOCS57990.2023.00037},
  doi          = {10.1109/FOCS57990.2023.00037},
  bibsource    = {dblp computer science bibliography, https://dblp.org}
}
